\documentclass[12pt, onecolumn]{IEEEtran}
\ifCLASSINFOpdf
\else
\fi
\usepackage{subfig}
\usepackage{array}
\usepackage{subfig}
\usepackage{amsmath}
\usepackage[nolist,nohyperlinks]{acronym}
\usepackage{graphicx,wrapfig,lipsum}
\usepackage{rotating}
\usepackage{amsmath, amssymb, amsthm}
\usepackage{stmaryrd}
\usepackage{tikz}
\usepackage{verbatim}
\usepackage{url}
\usepackage[utf8]{inputenc}
\usepackage[english]{babel}

\newtheorem{corollary}{Corollary}[]
\newtheorem{proposition}{Proposition}[]
\newtheorem{lemma}[]{Lemma}
\usepackage{multicol}
\usepackage{xr-hyper}
\usepackage{algorithm}
\usepackage{algpseudocode}
\usepackage{scalerel}
\usepackage{multicol}
\usepackage{setspace}
\newtheorem{definition}{Definition}
\usepackage[noadjust]{cite}
\usepackage{bbm}
\usepackage[normalem]{ulem}
\usepackage{color}
\usepackage{dsfont}

\usetikzlibrary{automata}
\usetikzlibrary{arrows}
\allowdisplaybreaks
\begin{document}
\bstctlcite{IEEEexample:BSTcontrol}
	%
	\title{Coverage Analysis and Load Balancing in HetNets with mmWave Multi-RAT Small Cells}
	\author{Gourab Ghatak$^{\dagger}$ $^\ddagger$, Antonio De Domenico$^{\dagger}$, and Marceau Coupechoux$^\ddagger$
 \\ \small{ $^{\dagger}$CEA, LETI, MINATEC, F-38054 Grenoble,
France; $^\ddagger$LTCI, Telecom ParisTech, Universit\'e Paris Saclay, France.}
\\ \small{Email: gourab.ghatak@cea.fr; antonio.de-domenico@cea.fr, and marceau.coupechoux@telecom-paristech.fr}}
		\maketitle
        \vspace{-2cm}
        \thispagestyle{empty}
	\begin{abstract}
We characterize a two tier heterogeneous network, consisting of classical sub-6GHz macro cells, and multi Radio Access Technology (RAT) small cells able to operate in sub-6GHz and millimeter-wave (mm-wave) bands. For optimizing coverage and to balance loads, we {propose a {two-step} mechanism based on} two biases for tuning the tier and RAT selection, {where the sub-6GHz band is used to speed-up the initial access procedure in the mm-wave RAT}. First, we investigate the effect of the biases in terms of signal to interference plus noise ratio (SINR) distribution, cell load, and user throughput. More specifically, we obtain the optimal biases that maximize either the SINR coverage or the user downlink throughput. Then, we characterize the cell load using the mean cell approach and derive upper bounds on the overloading probabilities. Finally, for a given traffic density, we provide the small cell density required to satisfy system constraints in terms of overloading and outage probabilities. Our analysis highlights the importance of deploying dual band small cells in particular when small cells are sparsely deployed or in case of heavy traffic.
	\end{abstract}
	%
	\IEEEpeerreviewmaketitle
	\section{Introduction}
	Future cellular networks will require a tremendous increase in data rates. This multi-fold {enhancement} cannot be achieved through incremental improvements on existing schemes \cite{andrews2014will}. For this, two techniques are particularly attractive: network densification using small cells \cite{lopez2015towards} and mm-wave wave communications~\cite{rappaport2013millimeter}. Densification of cellular networks consists of massive deployments of small cells, overlaying the existing macro cell architecture. Traditionally, small cells are deployed in sub-6GHz frequencies with the aim of offloading macro-cells.  This calls for Inter-Cell Interference Coordination ~\cite{okino2011pico, eguizabal2013interference} and load balancing~\cite{ali2015load}. To further increase the data rates, millimeter-wave (mm-wave) small cells, providing a very high bandwidth, are gaining popularity.
Apart from the large bandwidths, mm-wave communication comes with highly directional antennas, which greatly reduces the co-channel interference \cite{ghosh2014millimeter}. However, this technology is characterized by large path-loss and high sensitivity to blockages. {Because of the stronger path-loss, beamforming techniques should be deployed to mitigate it and this poses new issues in terms of coverage.} {Moreover, providing initial access to standalone mm-wave base stations using beamtraining with thin beams presents a difficult challenge~\cite{li2016initial}.
In this regard, the sub-6GHz band can be used to aid the initial access mechanism~\cite{mmMagic}. Specifically, given suitable signal processing mechanisms, the position and orientation of the users relative to a sub-6GHz BS can be determined (see e.g.,~\cite{kangas2013angle}). If sub-6GHz and mm-wave BS are co-located, or their position and orientation relative to one another are known, the coarse-grained angle information for beamtraining of the mm-wave RF front-end can be derived easily, which significantly speeds up the initial access procedure.}
As a result, it is unrealistic to assume ubiquitous coverage with only mm-wave small cells, and it is envisioned that multiple radio  access techniques (RATs) will co-exist in future cellular networks \cite{onoe20161}{~\cite{5783993}}. 
    
In this paper, we analyze the signal to interference plus noise ratio (SINR) distribution, the cell load and the downlink user throughput in a heterogeneous network with multi-RAT small cells using stochastic geometry. In order to optimize the user's SINR or to balance loads between tiers and RATs, we propose a cell association scheme based on two biases. In addition, we show the interest of deploying multi-RAT small cells to improve users' Quality of Service (QoS). 
 

\subsection{Related Work}
Elsawy et al., have presented a comprehensive survey on stochastic geometry to model multi-tier cellular networks \cite{elsawy2013stochastic}. The SINR and physical data rate distributions have been derived in the literature by {Bai et al.}~\cite{bai2015coverage} for single-tier mm-wave networks, by {Singh et al.}~\cite{singh2014joint} for multi-tier sub-6GHz and by {Di Renzo} for mm-wave networks~\cite{di2015stochastic}. 
In case of small cells operating in the same band of the macro cell, Singh et al.~\cite{singh2014joint}, have shown that{,} without advanced interference management techniques, the SINR decreases with increasing offloading bias. On the contrary, in this paper, we investigate how employing mm-wave in conjunction with sub-6GHz in small cell{s} affects the system performance, and we show that optimizing the offloading biases can increase the user's SINR.
 
{Omar et al.~\cite{omar2016performance} have considered separate mm-wave and sub-6GHz BS. They characterized the blockage in a suburban context using real data from the Lancaster university, UK. The results provided by the authors are greatly limited since they use simulation studies in a specific scenario. These results may not be applicable in other network architectures.}
{In the context of random networks, Yao et al.~\cite{yao2016coverage}}, similar to Di Renzo~\cite{di2015stochastic} have characterized the SINR coverage probability and {the} physical data rate {in} a multi-tier mm-wave network. However, {the authors have not studied how traffic dynamics in a multi-user scenario impacts} the network performance and the average user throughput. On the other hand, Elshaer et al.~\cite{elshaer2016downlink} have analyzed a multi-tier network {with sub-6GHz macro cells and mm-wave small cells}. {They have derived the SINR coverage probability as a function of the tier association bias, and they have shown only by simulations that a non-trivial optimal tier selection bias may exist. They have also investigated the relation between the association bias and throughput but without considering dynamic traffic.} Moreover, {they have characterized} the load {by} using the average number of associated users in a cell{;} although, for a more realistic characterization, a dynamic traffic model should be considered. Furthermore, they have not optimized the user throughput while considering SINR outage constraints as well as overloading constraints.
 
In this perspective, Bonald and Proutière~\cite{bonald2003wireless} have studied the relations between the traffic arrival rate and the cell load for a single cell scenario. In the case of single-tier cellular network, Blaszczyszyn and Karray~\cite{blaszczyszyn2014user} have approximated the cell load by a mean-cell approach to calculate {the} number of active users in a cell and the average {user} throughput. We leverage on these studies to design the optimal load balancing in multi-RAT heterogeneous networks and to derive bounds on overloading probabilities.

	\subsection{Contributions and Organization}

The contributions of this paper can be summarized as follows:
     
    \subsubsection{SINR Coverage in a multi-RAT Heterogeneous Network} By using stochastic geometry, we derive the association probabilities and the SINR distribution of a typical user in a multi-RAT heterogeneous network with small cells operating in sub-6GHz and millimeter wave bands. In the literature, SINR coverage and throughput analyses have not been performed so far for such a system model. 
    
    \subsubsection{Association Scheme for Tier and RAT Selection}We introduce a mechanism based on two biases, $Q_T$ and $Q_R$, for tuning the tier and RAT selection, respectively. The principle of using biased received power for association has been used so far for tier offloading, whereas in this paper, we introduce a second bias to distribute the users between the available RATs in the small cells. {Using these biases, we propose a two-step association scheme, in which initial access is performed in the sub-6GHz band. We compare our association mechanism with a more natural and exhaustive one-step association procedure in terms of sub-optimality of biased received power and downlink throughput.}
    {We show that this two-step association scheme fares better than cell association with beamtraining in mm-wave in terms of downlink throughput, specially in case of higher access delays.}
    
    \subsubsection{Bias Optimization for SINR Coverage} Contrary to single-RAT heterogeneous networks, biasing the received power can lead to an improved SINR in a multi-RAT system. However, bias optimization is difficult in general. In the general case, $Q_T$ and $Q_R$ can be obtained by brute force if the range of possible values is small. {To limit the complexity of this approach, we provide a strategy that sets $Q_R$ based on the ratio of the approximated mean SINR in sub-6GHz band and mm-wave. Thereafter, $Q_T$ is obtained using a random-restart hill-climbing algorithm with adaptive step-size.} We show that this strategy achieves near-optimal SINR coverage probability.
    We also highlight through simulations that sparse deployments require sub-6GHz band service for guaranteeing SINR requirements, whereas, in case of dense deployments, mm-wave may provide good SINR coverage, but with limited macrocell offloading. However, we show that, with large macrocell offloading, users at the edge of small cells, even in relatively dense deployments, need sub-6GHz band service to receive appreciable SINR coverage.
    
\subsubsection{Cell Load Characterization and Load Balancing} Next, we analyze the effect of traffic density on the downlink user throughput {by} using a M/G/1/PS queue model. {The existing literature in stochastic geometry defines the cell load as the average number of associated full buffer users, uniformly distributed over the cell area, see e.g.,~\cite{elsawy2013stochastic, singh2013offloading}. This approach is static in nature and ignores the effect of dynamic traffic on the user distribution: users with low data rate tend indeed to stay longer in the system so that the user distribution becomes inhomogeneous in space. To account for this effect, we rely on results from queuing theory~\cite{bonald2003wireless} and characterize the load of each cell by the mean cell approximation~\cite{blaszczyszyn2014user}. We solve a fixed point equation for the load to take the load of the interfering base stations into account.} Accordingly, we derive upper bounds on the probability for a cell in each tier and RAT to become overloaded. Based on the derived bounds, we provide values of minimum necessary deployment densities required for a given traffic density so as to limit overloading and outage. We then derive and optimize the downlink user throughput with respect to tier and RAT biases under these constraints. We analyze the fundamental trade-off between user throughput, overloading and outage probabilities. We finally highlight that the capability of the small cells to operate also in the sub-6GHz band plays a key role to restrict outage, thereby justifying our system model.
    

The rest of the paper is organized as follows.
    In Section \ref{sec:System Model}, we introduce our two-tier heterogeneous network model. In Section \ref{sec:CARS}, {we describe the proposed tier and RAT selection procedure and we derive the related association probabilities}.
{Then, in {S}ection \ref{CPandR}, we compute and optimize the network d}ownlink SINR distribution in terms of the tier and RAT selection biases. 
 In Section \ref{Sec:Load}, we characterize the load of the network and the downlink user throughput under a dynamic traffic model, and, hence, we design the load balancing such that the user performance is maximized. Simulation results are provided in Section VI. Finally, the paper concludes in section VII. Main notations used in this paper are shown in Table~\ref{tab:Param}. 
	\section{System Model}
	\label{sec:System Model}
	\subsection{Two-Tier Network Model}
\label{sec: Network Model}
    \begin{table}[!t]
    \small
	\centering
    \caption{Notations and System Parameters} 
\begin{tabular}{|c | c | c |}
	\hline  Notation& Parameter& Value\\
\hline
	\hline $\phi_M$, $\lambda_M$ & MBS process and density & $\lambda_M$ = 5 per sq. km. \\  
    	\hline $\phi_S$, $\lambda_S$ & SBS process and density & $\lambda_S$ = 5-200 per sq. km. \\  
	 \hline {$P_M$, $P_S$}& {MBS/SBS power} &  {46 dBm, 30 dBm}\\ 
	 \hline  {$\alpha_{tLr}$, $\alpha_{tNr}$}&  {Approximated LOS/NLOS path-loss exponents} &  {2, 4}\\
     \hline  $G_0$ & Maximum directivity gain with mm-wave antenna& 36 dB \\
     \hline  $N_0$& Noise power density & -174 dBm/Hz\\
     \hline  {$B_{\mu}$, $B_{mm}$} &  {Sub-6GHz/mm-wave bandwidth} &  {20 MHz, 1 GHz}\\
     \hline  $\sigma_{N,mm}^2$, $\sigma_{N,\mu}^2$& Noise power & $N_0 B_{mm}$, $N_0 B_\mu$\\
     \hline  {$d_{M}$, $d_{S}$} &  {MBS/SBS LOS ball radius} &  {200 m, 20 m}\\
	\hline  {$\theta$} & {Beamwidth} & {15 degrees}\\
	\hline 
\end{tabular}
\label{tab:Param}
\end{table}
Consider a two-tier network consisting of macro BSs referred to as MBSs, and small cell BSs referred to as SBSs. MBSs are deployed to guarantee continuous coverage to the users. On the contrary, multi-RAT SBSs locally provide high data rate by jointly exploiting sub-6GHz and mm-wave bands. We also assume that the same sub-6GHz band is shared by MBSs and SBSs. Therefore, users receiving services on this band experience both co-tier and cross-tier interference.
MBS and SBS locations are modeled as independent Poisson point processes (PPP), $\phi_M$ and $\phi_S$, with intensities $\lambda_M$ and $\lambda_S$, respectively. Let the transmit power of MBS be given by $P_M$; the small cell transmit power, in both the bands, is assumed to be equal to $P_S$. End users are assumed to be distributed according to a {PPP} $\phi_U$, independent of both $\phi_M$ and $\phi_S$. Due to the independence of the PPPs and Slivnyak's theorem \cite{elsawy2013stochastic}, without loss of generality, we carry out our downlink analysis considering a typical user located at the origin.

	\subsection{Blockage Processes}
	Cellular networks generally suffer from {link} blockages due to buildings, vehicles{,} etc. We assume a blockage process independent of the BS processes. Let the probability of a MBS and SBS to be in line of sight (LOS) with respect to the typical user {at a distance $r$}{,} be denoted by $p_M(r)$ and $p_S(r)$, respectively. For a given SBS, the LOS probability in  sub-6GHz is assumed to be the same as that in mm-wave. This is because, the probability of a signal to be blocked mainly depends on the  blockage process, which is independent of the carrier frequency \cite{bai2014analysis}.
	Due to the blockages, MBSs and SBSs can be categorized into either LOS or NLOS (non line of sight) processes: {$\phi_{ML}$, $\phi_{MN}$, $\phi_{SL}$, and $\phi_{SN}$, respectively}. The intensity of these modified processes are given by $p_M(r)\lambda_M$, $(1 - p_M(r))\lambda_M$, $p_S(r)\lambda_S$, and $(1 - p_S(r))\lambda_S${,} respectively.
{In our work, we use the LOS ball approximation introduced in \cite{bai2015coverage}. Accordingly, let $d_M$ be the MBS LOS ball radius. The probability of the typical user to be in LOS from a MBS at a distance $r$ is $p_M(r)=1$, if $ r < d_M$, and $p_M(r)=0$, otherwise.  
	We assume a similar LOS ball for the SBS process with a different radius $d_S$.
	
	\subsection{Directional Beamforming in mm-wave}
	In case of mm-wave operations, the {received powers take advantage of} the directional antenna gain of the transmitter and the receiver. {The user and the serving BS are assumed to be aligned, whereas the interfering BSs are randomly oriented with respect to the typical user.}
{Here, we assume a tractable model, where the product of the transmitter and receiver antenna gains, $G$, takes on the values $a_k$ with probabilities $b_k$ as given in Table 1 of \cite{bai2015coverage}. Let the maximum value of $G$ be $G_0$. }
	\subsection{Path-loss Processes} \label{subsec:pathloss}
    We assume a distance based path-loss model where the path-loss at a distance $d_{tvr}$ from a transmitter is given by: $l_{tvr}(d) = K_{tvr}d_{tvr}^{-\alpha_{tvr}}$ for a BS of type $tvr$, i.e., characterized by tier $t$ (MBS or SBS), visibility state $v$ (LOS or NLOS), and RAT $r$ (sub-6GHz or mm-wave). Parameters $K_{tvr}$ and $\alpha_{tvr}$ are derived from 3GPP UMa model for sub-6GHz MBSs, Umi model for sub-6GHz SBSs \cite{36.814}, and Umi model for mm-wave data transmission in SBSs~\cite{38.900}. By assuming a fast fading that is Rayleigh distributed with variance equal to one, the average received power is thus given by $P_{tvr} = P_tK_{tvr}d_{tvr}^{-\alpha_{tvr}}$, where $P_t$ is the transmit power of a BS of tier $t$.

With our values (see Table \ref{tab:Param}) of transmit powers, path-loss exponents, and LOS ball radii, we have 
	$
	\frac{d_S^{\alpha_{SL\mu}}}{K_{SL\mu}P_S} \le \frac{d_M^{\alpha_{ML\mu}}}{K_{ML\mu}P_M} \le \frac{d_S^{\alpha_{SN\mu}}}{{K_{SN\mu}}P_S} \le \frac{d_M^{\alpha_{MN\mu}}}{K_{MN\mu}P_M}.
	$
 The analysis in this paper is done considering that this ordering does not change even when powers are biased\footnote{{This assumption of ordering is considered only for the sake of simplicity and practicality. 
Considering higher bias values marginally alters the theoretical developments by modifying integral bounds in association probabilities. From an engineering point of view, very high bias values also lead to unacceptable outage probabilities and thus are of little interest.}}.
This assumption is reasonable considering that if a LOS BS exists and the tier bias is moderate, its biased received power is very likely to be greater than that of any NLOS BS. {Accordingly, we analyze the performance of the network with tier-selection bias  $(Q_T)$ in the range: $1 \leq Q_T \leq \frac{d_S^{\alpha_{SN\mu}}K_{ML\mu}P_M }{d_M^{\alpha_{ML\mu}}K_{SN\mu}P_S}$ = $Q_T^{max}$.}
\subsection{{Dynamic Traffic Model}}
{We consider a model in which users arrive in the system, download a file, and leave the system. Any new download by the same user is considered as a new user. The arrival process of the new users is Poisson distributed with an intensity $\lambda$ [users $\cdot$s$^{-1}\cdot$m$^{-2}$] and these new users are uniformly distributed over the network area $A$. The average file size is $\sigma$ [bits/user]. When there are $n$ users simultaneously served by a base station, the available resources are equally shared between them in a Round Robin fashion. Accordingly, we define the traffic density $w$ in the network as $w = \lambda \cdot \sigma$ [bits$\cdot$s$^{-1}\cdot$m$^{-2}$]. Note that, while the user arrivals are uniform in space, as the space-time process evolves, users farther from the serving base stations which are characterized by lower data rates stay longer in the system, resulting in an inhomogeneous distribution of active users in the network.}

\section{Cell Association Procedure}
In this section, we propose a cell association scheme based on tier and RAT selection biases and we derive the corresponding association probabilities. We start below by a preliminary result. 

\subsection{Distribution of {the Path-loss Process}}
To analyze the cell association, {path-loss processes are reformulated} as one dimensional processes,
	$\phi'_{tvr} = \{\xi_{tvr,i}: \xi_{tvr,i} = \frac{||x_i||^{\alpha_{tvr}}}{K_{tvr}P_t} , x_i \in \phi_{tv}\}$, $t\in \{M,S\}$, $v\in \{L,N\}$, $r\in \{\mu, m\}$.
	The processes $\phi'_{tvr}$ are non-homogeneous with intensities calculated as below. 
	
	\begin{lemma} The intensity measures of the LOS and NLOS {path-loss} processes, $\phi'_{tLr}$ and $\phi'_{tNr}$ are:
	\begin{align}
	\Lambda'_{tLr}(0,x) &= \begin{cases}
	\pi\lambda_t{(K_{tvr}P_t)}^{\frac{2}{\alpha_{tLr}}}x^{\frac{2}{\alpha_{tLr}}}, & x < \frac{d_t^{\alpha_{tLr}}}{{K_{tvr}P_t}} \nonumber \\
	\pi\lambda_td_t^{{2}}, & x > \frac{d_t^{\alpha_{tLr}}}{{K_{tvr}P_t}}
	\end{cases}, \\
   \Lambda'_{tNr}(0,x) & = \begin{cases}
	0, & x < \frac{d_t^{\alpha_{tNr}}}{{K_{tvr}P_t}} \\
	\pi\lambda_t({(K_{tvr}P_tx)}^{\frac{2}{\alpha_{tNr}}} - d_t^2), & x > \frac{d_t^{\alpha_{tNr}}}{{K_{tvr}P_t}}
	\end{cases}.
    \label{eq:IntensityMeasure}
	\end{align}
    \end{lemma}
	\begin{proof}
		The derivation of the intensity measure is similar to that in \cite{zhang2014stochastic}. 
	\end{proof}
	The related intensities are obtained by differentiating the intensity measures, and are given by:	
	\begin{align}
	\lambda'_{tLr}(x) &= \begin{cases}
	\frac{2\pi\lambda_t {(K_{tvr}P_t)}^{\frac{2}{\alpha_{tLr}}}}{\alpha_{tLr}}x^{\frac{2}{\alpha_{tLr}}-1} ,  &x < \frac{d_M^{\alpha_{tLr}}}{{K_{tvr}P_t}} \nonumber  \\ 
	0, & x > \frac{d_t^{\alpha_{tLr}}}{{K_{tvr}P_t}}  
	\end{cases}\\
	\lambda'_{tNr}(x) &= \begin{cases}
	0, & x < \frac{d_M^{\alpha_{tNr}}}{{K_{tvr}P_t}} \\ 
	\frac{2\pi\lambda_t {(K_{tvr}P_t)}^{\frac{2}{\alpha_{tNr}}}}{\alpha_{tNr}}x^{\frac{2}{\alpha_{tNr}}-1} , &  x > \frac{d_t^{\alpha_{tNr}}}{{K_{tvr}P_t}}.  
	\end{cases}
    \label{eq:Intensity}
	\end{align}
	\begin{lemma} \label{lemma:firstpoint1D}
The probability density function (pdf) of the first point of $\phi'_{tv\mu}$, which corresponds to strongest sub-6GHz BS, is:
$$
f_{\xi_{tv\mu1}}(r) = e^{-\Lambda'_{tv\mu}(0,r)} \lambda'_{tv\mu}(r).
$$
\end{lemma}
	\begin{proof}
		The pdf of the first point in $\phi'_{tv\mu}$ is computed as 
        $$
		f_{\xi_{tv\mu1}}(r)  = \frac{d}{dr} \left[ \mathbb{P}(\phi'_{tv\mu} \cap (0,r) = 0) \right] = \frac{d}{dr} \left[e^{-\Lambda'_{tv\mu}(0,r)}\right]  = e^{-\Lambda'_{tv\mu}(0,r)} \lambda'_{tv\mu}(r),
        $$
        where $\Lambda'_{tv\mu}$ and $\lambda'_{tv\mu}$ are given by Eq. \eqref{eq:IntensityMeasure} and Eq. \eqref{eq:Intensity}, respectively.
        	\end{proof}

\subsection{Tier and RAT Selection Scheme}

For the cell association mechanism, we assume that BSs send their control signals in the sub-6GHz band. This is due to the fact that sub-6GHz communication  benefits from a higher reliability and better coverage than mm-wave signals~\cite{shokri2015millimeter}. Our scheme is based on two biases $Q_T$ and $Q_R$ for selecting the tier and the RAT respectively, to which the user will be associated. Parameter $Q_T$ is the classical cell range expansion parameter~\cite{singh2014joint}: a user compares the strongest MBS signal with the strongest {\it biased} SBS signal. By varying $Q_T$, we are able to offload users from MBSs to SBSs. Once associated to a SBS, in our approach, a user compares the sub-6GHz received signal with the mm-wave signal strength {\it biased} with a second parameter $Q_R$. By varying $Q_R$, users can be distributed between RATs of the same SBS\footnote{An alternative association scheme could be realized through the control of the SBS power in the different bands. However, as the transmit powers of SBSs are generally limited, we do not take this into consideration. Moreover, our approach can be easily adapted to study this alternative scheme.}. 
 The association policy, summarized in Algorithm 1, consists of two steps: tier selection and RAT selection. 
	\label{sec:CARS}
\subsection{Tier Selection}
	\label{subsec:CA}   	
The tier selection is {based on the transmitted {signal on the sub-6GHz band}.} As a result, a user can be served either by: 1. an MBS in LOS  (ML), 2. an MBS in NLOS (MN), 3. an SBS in LOS (SL), or 4. an SBS in NLOS (SN). The biased received powers in sub-6GHz from the strongest LOS MBS, NLOS MBS, LOS SBS, and NLOS SBS are denoted as $P_{ML\mu1}, P_{MN\mu1}, Q_TP_{SL\mu1}$, and $Q_TP_{SN\mu1}$, respectively. {User association is only based on measured biased received power}. With the ordering assumption of Section II-D, however, a user associates with an NLOS BS only in absence of an LOS BS.
\begin{algorithm*}[!htp] 
\small
    {
Algorithm 1: Tier and RAT Selection \\
		\begin{algorithmic}[1]
\vspace*{0.1cm}
			\State Measure downlink sub-6GHz received powers from all MBS, SBS.
            \State Let $P_{Mv\mu 1}$ and $P_{Sv\mu 1}$ be the strongest powers received from an MBS and an SBS, resp.            
            \If {$P_{Mv\mu 1} \geq Q_T P_{Sv\mu 1}$}
            	\State Associate to the strongest MBS
            \Else
            	\State Associate to the strongest SBS
            	\State Measure the mm-wave received power from the SBS $(P_{Svm 1})$.
            	\If{{$P_{Sv\mu 1} \geq Q_R P_{Svm 1}$}}
            		\State Start service from SBS in {sub-6GHz} band.
            	\Else
            		\State Start service from SBS in mm-wave band.
            	\EndIf
            \EndIf         
		\end{algorithmic}
        }
        \label{ALGO}
\end{algorithm*}
{It must be noted that the user does not know the visibility state of the base stations and associates only according to the biased received powers. The result that the user associates with an NLOS BS only in the absence of a LOS BS follows from the ordering described in Section II-D, which in turn, is a result of the values of the transmit powers and LOS ball radii.}  As a consequence, for a LOS BS, the association probability of a typical user with tier $t$ can be calculated as:
\begin{align}
\mathbb{P}_{tL} = \mathbb{E}\left[\mathds{1}(tL)\right]\cdot\mathbb{E}[\mathds{1}(t'L)]\cdot\mathbb{P}(\tilde{Q}_{T}P_{tL\mu1}>\tilde{Q}_{T'}P_{t'L\mu1}) + \mathbb{E}\left[\mathds{1}(tL)\right]\cdot(1-\mathbb{E}[\mathds{1}(t'L)]),
\label{eq:LOSASSO}
\end{align}
where $t, t' \in \{M,S\}$, $t \neq t'$, and $\mathds{1}(.)$ is an indicator function: $\mathds{1}(tL) = 1$ if and only if a point of tier $t$ with visibility state $L$ exists. The value of $\tilde{Q}_{T}$ is equal to 1 if $t= M$, else it is equal to $Q_T$. The first term of Eq. \eqref{eq:LOSASSO} is the product of the probabilities of 1) the existence of a LOS SBS and 2) the existence of a LOS MBS and 3) that the received power from the serving tier is greater than the one from the non-serving tier. The second term is the product of the probabilities of the existence of at least one LOS BS of the serving tier and the absence of a LOS BS of the non-serving tier.
In the same way, for the NLOS BSs, we have:
\begin{eqnarray}
\mathbb{P}_{tN} = (1-\mathbb{E}[\mathds{1}(t_ML)])\cdot(1-\mathbb{E}[\mathds{1}(t_SL)])\cdot\mathbb{P}(\tilde{Q}_{T}P_{tN\mu1}>\tilde{Q}_{T}P_{t'N\mu1}).
\label{eq:NLOSASSO}
\end{eqnarray}
From these observations, we can deduce the tier selection probabilities as follows. 

	\begin{lemma}The tier selection probabilities are:
    \begin{align}
&\mathbb{P}_{ML} = \exp(-\pi \lambda_M d_M^2)\cdot\exp(-\pi \lambda_S d_S^2) \cdot W_1 + \exp(-\pi \lambda_M d_M^2)\cdot \left(1 - \exp(-\pi \lambda_S d_S^2)\right),\nonumber \\
&\mathbb{P}_{MN}  = \left(1-\exp(-\pi \lambda_M d_M^2)\right)\cdot \left(1-\exp(-\pi \lambda_S d_S^2)\right) \cdot W_2, \nonumber \\
&\mathbb{P}_{SL}   = \exp(-\pi \lambda_M d_M^2)\cdot\exp(-\pi \lambda_S d_S^2) \cdot(1- W_1) + \exp(-\pi \lambda_S d_S^2)\cdot \left(1 - \exp(-\pi \lambda_M d_M^2)\right),\nonumber \\
&\mathbb{P}_{SN}  = \left(1-\exp(-\pi \lambda_M d_M^2)\right)\cdot \left(1-\exp(-\pi \lambda_S d_S^2)\right) \cdot (1-W_2),\nonumber
\end{align}
  where,
  \begin{align} 
W_1 &=  \frac{(1 - e^{-(K_1 + 1)t_1})}{1 + K_1}  + \exp(-\pi \lambda_S d_S^2)\left[ \exp\left(-\Lambda'_{ML\mu}\left(0,\frac{d_S^{\alpha_{SL\mu}}}{Q_TK_{SL\mu}P_S}\right)\right) - \exp(-\pi \lambda_M d_M^2) \right] \label{eq: rmlgrsl},
	 \nonumber \\
W_2 &= \exp(-\pi \lambda_S d_S^2)  \frac{e^{-(K_2 + 1)t_2}}{1 + K_2},\nonumber
\end{align}
 $K_1 = \pi \lambda_S (\frac{K_{SL\mu}P_SQ_T}{K_{ML\mu}P_M})^{\frac{2}{\alpha_{SL\mu}}}(\pi\lambda_M)^{-\frac{\alpha_{ML\mu}}{\alpha_{SL\mu}}}$, $t_1 = \pi \lambda_M (K_{ML\mu}P_M)^{\frac{2}{\alpha_{ML\mu}}}\left(\frac{d_S^{\alpha_{SL\mu}}}{Q_TK_{SL\mu}P_S}\right)^{\frac{2}{\alpha_{ML\mu}}}$,
\\
$K_2 = \pi \lambda_S (\frac{K_{SN\mu}P_SQ_T}{K_{MN\mu}P_M})^{\frac{2}{\alpha_{SN\mu}}}(\pi\lambda_M)^{-\frac{\alpha_{MN\mu}}{\alpha_{SN\mu}}}$, and $t_2 = \pi\lambda_Md_M^2(K_{MN\mu} P_M)^{\frac{2}{\alpha_{MN\mu}}-1}$.

	\end{lemma}
    \begin{proof}
		See Appendix \ref{App:TierSel}.
	\end{proof}

	\begin{lemma} Given that a user is associated to a tier $t$ of visibility state $v$, the pdf of the point in the 1D process of the serving BS is given by:
        \begin{equation}
\hat{f}_{\xi tv\mu1}(x) = \frac{{f}_{\xi tv\mu1}(x)}{\mathbb{P}_{tv}}\prod_{\forall (t'v' \neq tv)}\mathbb{P}(\phi'_{t'v'} \cap (0,x) = 0),
\end{equation}
where ${f}_{\xi tv\mu1}(x)$ is given by Lemma~\ref{lemma:firstpoint1D}.
\end{lemma}
\begin{proof}
The proof follows from Lemma 3 above and Lemma 3 of \cite{bai2015coverage}.
\end{proof}
	\subsection{RAT Selection in SBS}
	\label{subsec:RS}
	A dual-band user, associated with an SBS, is served using mm-wave if and only if the biased estimated power in the mm-wave band is larger than the power received in the sub-6GHz band.
	
	\begin{lemma}  Given that a user is associated with a SBS of visibility state $v$, the sub-6GHz and mm-wave RAT selection probabilities are respectively given by:
	\begin{align}
	\mathbb{P}_{v\mu} &=\exp\left(-\pi\lambda_S\left(\frac{K_{Svm}G_0Q_R}{K_{Sv\mu}}\right)^{\frac{2}{\alpha_{Svm} - \alpha_{Sv\mu}}}\right)  \label{PSmu}\\
	\mathbb{P}_{vm} &= 1 - \mathbb{P}_{v\mu}.
	\end{align}
    \end{lemma}
	\begin{proof}
		See Appendix \ref{App:RATSel}.
	\end{proof}

We denote $\mathbb{P}_{tvr}\triangleq\mathbb{P}_{tv}\mathbb{P}_{vr}$ as the association probability to a BS of type $tvr$ with the convention that when $t=M$, $\mathbb{P}_{v\mu}=1-\mathbb{P}_{vm}=1$.

\subsection{Comparison to a One-Step Association Strategy}
{
It must be noted that our proposed two-step association scheme is different from a more natural and exhaustive scheme (e.g., \cite{elshaer2016downlink}), which directly compares the biased received powers from all the tiers and RATs (i.e., one-step procedure). In this regard, our two-step association procedure suffers from some sub-optimality with respect to the biased received power. {However, access delay is lower with our strategy because the users position and orientation can be acquired in the sub-6GHz band before performing beamtraining.} 
}

{First, we show that both the one-step strategy and our approach result in the same RAT selection, given that the user associates with the small cell tier. Then, our strategy differs from the one-step strategy when a user associates to an MBS while the biased power received from an SBS in mm-wave is higher than the biased power received from the MBS. We characterize hereafter the probability of this event.}
{\begin{proposition}
If the typical user receives a higher sub-6GHz received power from an SBS $S_1$ as compared to an SBS $S_2$, then it also receives higher mm-wave power from $S_1$ than from $S_2$. Moreover, the tier selection and RAT selections biases $Q_T$ and $Q_R$, do not impact this ordering of received powers.
\end{proposition}
\begin{proof}
See Appendix \ref{App:Same}.
\end{proof}
}
{
From Proposition 1, we conclude that it is not possible for the typical user to have a higher received power in sub-6GHz band from SBS $S_1$ as compared to $S_2$ and lower mm-wave power from the same. Thus, the two schemes result in the same RAT selection, in case the user associates with the SBS tier.}
{
Therefore, the only difference in association arises when the biased received power from the strongest SBS (denoted $S_1$) in sub-6GHz band is less than that received from the strongest MBS (denoted by $M_1$), while simultaneously, the biased received power from $S_1$ in mm-wave is higher than the biased received power from $M_1$. Let us call these events $E_1$ and $E_2$, respectively. This results in sub-optimal association of some users in the sense that these users are not associated to the tier-RAT pair providing the highest {biased} power. We have the following result to model this sub-optimal association. }
{
\begin{lemma}
The probability of suboptimal association in case of an association with LOS MBS instead of mm-wave LOS SBS is given as:
\begin{align}
 \mathbb{P}_{SO} =2\pi\lambda_M \frac{1 - \exp\left(-\pi\left(\lambda_S\zeta_2 - 2\lambda_S \zeta_1 + \lambda_M\right) d_M^2\right)}{2\zeta_2},
 \label{eq:sub_opt}
\end{align}
where $\zeta_1 = \frac{P_S Q_T}{P_M} $ and  $\zeta_2 = \frac{K_m P_S Q_R Q_T}{K_{\mu}P_M}$
\end{lemma}
}
{\begin{proof}
See Appendix \ref{App:Sub_OPT}.
\end{proof}
In Section VI, we provide numerical results to show that the sub-optimality is limited, and this loss can be compensated by a faster access procedure, which may increase the network throughput.}

\subsection{A Simple Strategy to Prioritize mm-Wave RAT}

Depending on the network load and the active services, the mobile operator may want to prioritize one RAT over the other. For instance, the utilization of mm-wave frequencies for latency sensitive applications, can be an attractive strategy to offload the sub-6GHz band, which can mainly be dedicated to communications requiring reliability and continuous service. In the following, we propose a strategy to achieve this goal. For that, we introduce the following definition:

\begin{definition}
The critical distance with respect to the typical user is the distance of the SBS from which the typical user receives equal mm-wave and sub-6GHz power.
\end{definition}

For our system model, the critical distance for the LOS SBS tier can be expressed as:
\begin{align}
d_{CL} = \left(\frac{K_{SLm}}{K_{SL\mu}}G_0\right)^{\frac{1}{\alpha_{SLm}-\alpha_{SL\mu}}}
\label{eq:Critd}
\end{align}
\begin{proposition}
If there exists exactly one point of the LOS SBS process within the critical distance, the typical user always selects mm-wave as serving RAT. Moreover, in this scenario, this is the optimal strategy in terms of SINR for the typical user.
\end{proposition}
\begin{proof}
In the case where this condition holds, the useful signal received in mm-wave is greater than that received in sub-6GHz (as per definition of $d_{CL}$). Thus, the typical user always selects the mm-wave RAT from the serving SBS. Moreover, as all interfering LOS SBS are outside $d_{CL}$, the sub-6GHz interference has state-wise dominance with respect to the mm-wave interference. Hence, the mm-wave SINR is always larger than the sub-6GHz SINR.
\end{proof}
From Eq. \eqref{eq:Critd}, we see that, for given path-loss exponent values of each user, the critical distance can be controlled by varying the product of the transmitter and receiver antenna gain $G_0$. This enables the users served by LOS SBSs to adjust their antenna gain in order to select mm-wave communications, and ensure that this choice is optimal from the SINR perspective. 
In addition, given a fixed antenna gain $G_0$, we have the following corollary, which provides the deployment density of SBSs that maximizes the probability of occurrence of a single LOS SBS within the critical distance. 
\begin{corollary}
The maximum probability of occurrence of exactly one point of LOS SBS within the critical distance is $1/e$, and this occurs at: $$\lambda_S = \frac{1}{\pi}\left(\frac{K_{Svm}}{K_{Sv\mu}}G_0\right)^{\frac{2}{\alpha_{SL\mu} - \alpha_{SLm}}}.$$
\end{corollary}
\begin{proof}
The probability of existence of only one point within the critical distance is calculated as:

$\mathbb{P}\left(\phi'_{SL} \cap b(0,d_{CL}\right) = 1) = \pi\lambda_Sd^2_{CL}\exp(-\pi\lambda_Sd^2_{CL}),$
where $b(0,d_{CL})$ is the ball of radius $d_{CL}$ centered at the origin.
The maximum value of this probability occurs at $\pi\lambda_Sd^2_{CL} = 1$, then substituting the value of $d_{CL}$ from Eq. \eqref{eq:Critd} completes the proof.
\end{proof}
	\section{Downlink SINR Distribution}
	\label{CPandR} 
	In this section, we first derive the downlink SINR coverage probability for the maximum biased received power association policy and then optimize the biases with respect to the cell coverage.
\subsection{SINR Coverage Probability}
	The SINR coverage probability at a threshold $\gamma$, can be expressed as
	$
	\mathbb{P}_C(\gamma) = \mathbb{P}(SINR>\gamma).
	$
Following the theorem of total probabilities, we have:
	\begin{equation}
	\mathbb{P}_C(\gamma) = \!\!\!\!\!\!\!\!\!\!\!\!\sum\limits_{t\in \{M, S\},\;v\in \{L,N\},\;r\in \{\mu, m\}}\!\!\!\!\!\!\!\!\!\!\!\!\mathbb{P}(SINR_{t,v,r}>\gamma|t,v,r)\mathbb{P}_{tvr},
    \label{eq:CovProb}
	\end{equation}
We divide the problem of finding the overall coverage probability into two parts: the one related to the sub-6GHz service and the one associated with the mm-wave service, and we compute the coverage probability by relying on 1D processes $\phi'_{tvr}$.
	\begin{lemma} The conditional SINR coverage probability, given that the user is associated with a sub-6GHz BS of tier $t$ and visibility state $v$, is given by: 
	\begin{align}
	\mathbb{P}_{Ctv\mu}(\gamma) &= \int_0^\infty \exp{\left(-\gamma\cdot \sigma_{N,\mu}^2 \cdot x- \sum_{t',v'} A_{t'v'}(\gamma,x)\right)} \hat{f}_{\xi{tv\mu1}}(x) dx \label{MLCovP},\\
\mbox{where,\quad\quad\quad}	A_{t'v'}&= \int\limits_{l_{t'}}^\infty\frac{\gamma x}{y + \gamma x}  \Lambda'_{t'v'\mu}(dy)\nonumber, \quad \forall \; t' \in \{M,S\}, v' \in \{L,N\}.
    \end{align} 
Additionally, $l_{t'}=x$ if $t'=t$, $l_{t'}=Q_{T}\cdot x$, when $t=M$ and $t'=S$, and $l_{t'}=x/Q_T$, when $t=S$ and $t'=M$.
 \end{lemma}
	\begin{proof}
		See Appendix \ref{App:ProofCovP}.
	\end{proof}

	\begin{lemma} The conditional SINR coverage probability, given that the user is associated with a SBS in mm-wave of visibility state $v$, is given by:
	\begin{align}
	\mathbb{P}_{CSvm}(\gamma) &= \int_0^\infty \exp\left({-\frac{\gamma \cdot x \cdot \sigma_{N,mm}^2}{G_0} - B_1(\gamma,x) - B_2(\gamma,x)}\right)  \hat{f}_{\xi{Svm1}}(x) dx, \label{eq:CoVPmmwave} \\ 
   \mbox{with\quad} B_1(\gamma,x) &= \sum_{k=1}^4 \left(-b_k \int\limits_x^\infty\left(\frac{a_k\gamma x}{y + a_k\gamma x} \Lambda'_{Svm}(dy) \right)\right), \nonumber \\
  \mbox{and,\quad}  B_2(\gamma,x) &= \sum_{k=1}^4  \left(-b_k \int\limits_x^\infty\left(\frac{a_k\gamma x}{y + a_k\gamma x} \Lambda'_{Sv'm}(dy) \right)\right) \nonumber .
	\end{align}
    \end{lemma}
	\begin{proof}
		The proof follows in a similar way to that of {Lemma 7}.
	\end{proof}
	\subsection{A Near-Optimal Strategy for Bias Selection}

\label{sec:bias selection}
{On the one hand, obtaining optimal biases with respect to the SINR coverage probability is difficult because of the complex expressions. On the other hand, using brute force to search through all the possible pairs of tier and RAT selection biases can have a very high time-complexity which limits practical implementation. Accordingly, in this section, we propose a strategy to select the tier and RAT selection biases with the aim of maximizing the SINR coverage.}

{Specifically, the proposed strategy is based in two parts: 1) computing the RAT selection bias, $Q_R$ and 2) obtaining the tier selection bias $Q_T$ based on a random-restart hill-climbing algorithm.}

\subsubsection{Heuristic for Selection of $Q_R$}
{The heuristic to set the RAT selection bias $Q_R$ consists of computing the ratio of the mean signal to interference and noise perceived by the typical user on the sub-6GHz and mm-wave bands. That is:
\begin{align}
 Q_R = \frac{\mathbb{E}\left[\frac{S_{mm}}{I_{mm} + \sigma^2_{N,mm}}\right]}{\mathbb{E}\left[\frac{S_{\mu}}{I_{\mu} + \sigma^2_{N,\mu}}\right]},
 \label{eq:SINRratio}
\end{align} 
 where $I_\mu$ and $I_{mm}$, respectively, are the sum of the interference from all the (LOS and NLOS) BSs in sub-6GHz and mm-wave, respectively. It must be noted that evaluating the above expression without the knowledge of the coverage probability is not possible. However, with a relaxation of independence of the useful signal and the interference for each of the RATs, the expected values can be approximated using the results of~\cite{ghatak2016performance}. Once $Q_R$ is computed, $Q_T$ can be obtained by the following step}
\subsubsection{Random-Restart Hill-Climbing Algorithm for Selection of $Q_T$}
{We start with a random value of $Q_T$, i.e., $Q_T^0$ and calculate the gradient of $\mathbb{P}_C$ at $Q_T^0$. In case the gradient is non-negative, we increase the value of $Q_T$ by a step size of $k$. If the gradient is negative, we decrease the value of $Q_T$ by the same step size $k$. We continue this procedure with the updated value of $Q_T$ until the variation in $Q_T$ is sufficiently small. In the case where the product of two consecutive values of the gradient is non-positive, and as a result we cross a stationary point, we reduce the step size by a factor $\beta$ and continue the algorithm.}

{In our algorithm, $Q_T^{max}$ is the maximum value of the bias in the moderate range. {If the coverage probability is monotonic, quasi-convex or quasi-concave, this procedure provides the optimal value of $Q_T$. In the general case, the procedure stops at a local maximum in the range $1\leq Q_T\leq Q_T^{max}$.}
{This local maximum can be improved by repeatedly starting the same algorithm with random starting points}. This procedure to obtain $Q_T$ is summarized in Algorithm 2.} {In Section VI, we compare the performance of this proposed scheme with the optimal case.}

\begin{algorithm}
\small 
{Algorithm 2: Random-restart hill-climbing algorithm with Adaptive Step-Size\\
	\begin{algorithmic}[1]
    	\State Set $t = 1$, $k>0$, $\epsilon>0$ and $\beta>1$. 
		\State Set $Q_T(0) = Q_T^0$.
        \While{$|Q_T(t)-Q_T(t-1)|>\epsilon$}
        	\If {$\frac{d\mathbb{P}_C}{dQ_T}\left(Q_T(t)\right) > 0$}
        		\State $Q_T(t) = \min\{Q_T(t-1) + k,Q_T^{max}\}$.
            \Else
            	\State $Q_T(t) = \max\{Q_T(t-1) - k,1\}$.
        	\EndIf
        	\If{$\frac{d\mathbb{P_C}}{Q_T}\left(Q_T(t)\right)\cdot \frac{d\mathbb{P_C}}{Q_T}\left(Q_T(t-1)\right) < 0$}
        		\State $k \leftarrow \frac{k}{\beta}$.
        	\EndIf
            \State $t\leftarrow t+1$
        \EndWhile
\end{algorithmic}}
\end{algorithm}
\section{Cell Load, User Throughput, and Load Balancing}
\label{Sec:Load}
In the previous section we have focused only on coverage aspects, we now take into account cell loads to show how tier and RAT selection biases can improve the user average throughput.  
For this, we consider a multi-user system where the users share the available radio resources according to a round robin policy. 
\subsection{Cell Load Characterization}
{
According to our model of the traffic arrival process, the traffic density is given as ${w}= \lambda \cdot \sigma$ [bits/s/m$^2$]. For a single cell scenario, Bonald et al.~\cite{bonald2003wireless} have modeled the load of the cell of area $A$ as $\rho =   \int_{A} \frac{{w}}{R(s)} ds 
$ \cite{bonald2003wireless},
where $R(s)$ is the physical data rate of a user located at $s$.}
{In case of Poisson-Voronoi cells, the average load is generally difficult to evaluate because of the randomness in the shape and sizes of the cells. Furthermore, in a multi-cell scenario, the load of a cell depends on the SINR characteristics of the cell, which in turn, depends on the load of the other cells in the network. }

{We know from the ergodicity of the PPP, that the fraction of the BS of type $tvr$ that are idle is equal to the fractional idle time of the typical BS of same type. Accordingly, assuming that the load of the typical BS of type $tvr$ is given by $\bar{\rho}_{tvr}$, then, the fraction of idle BS of type $tvr$ is given by $1 - \bar{\rho}_{tvr}$.} 
{We substitute this value $\forall \; t,v,r$ in the calculation of the load as:
\begin{align}
\bar{\rho}_{tvr} = \int_{\gamma} \frac{{w}A_{tvr}}{B_{r}\log_2(1 + \gamma)}p_{tvr}(\bar{\rho},\gamma) d\gamma,
\end{align}
where the pdf of the SINR $p_{tvr}(\bar{\rho}_{tvr},\gamma)$ is a function of the average idle fraction of the BS and $\bar{\rho}$ is a vector of the fraction of idle BSs of all BS types, i.e., $\bar{\rho} = [\rho_{tvr}]$ $\forall \; t, v, r$.}
{This fixed point equation is then solved in an iterative manner to obtain the actual load of the BS of all the tiers (starting from zero load). {The, the} SINR coverage probability with $1 - \bar{\rho}$ fraction of BSs idle, given that the user is associated with a sub-6GHz BS of tier $t$ and visibility state $v$, is given by: 
	\begin{align}
	\mathbb{P}_{Ctv\mu}(\bar{\rho},\gamma) &= \int_0^\infty \exp{\left(-\gamma\cdot \sigma_{N,\mu}^2 \cdot x- \sum_{t',v'} A_{t'v'}(\gamma,x,\rho_{t'v'\mu})\right)} \hat{f}_{\xi{tv\mu1}}(x) dx \label{MLCovP},\\
\mbox{where,\quad\quad\quad}	A_{t'v'}(\gamma,x,\rho_{t'v'\mu})&= \int\limits_{l_{t'}}^\infty\frac{\gamma x}{y + \gamma x}  \Lambda'_{t'v'\mu}(dy,\rho_{t'v'\mu})\nonumber, \quad \forall \; t' \in \{M,S\}, v' \in \{L,N\}.
    \end{align} 
Additionally, $l_{t'}=x$ if $t'=t$, $l_{t'}=Q_{T}\cdot x$, when $t=M$ and $t'=S$, and $l_{t'}=x/Q_T$, when $t=S$ and $t'=M$.
The intensity measures $\Lambda_{tvr}$ are obtained by modifying $\lambda_t$ to $\lambda_t \rho_{tvr}$ for each BS type. The calculation for the mm-wave BS follows in the same way.}

It should be noted that in case of a Poisson distributed network, there exists a non-zero fraction of unstable cells ($\rho\geq1$), which cannot handle their load.   
\begin{lemma}
The probability of a typical cell of type $tvr$ to be unstable is bounded as:
\begin{align}
\mathbb{P}\left(\rho_{tvr} \geq 1\right) \leq \min \left\lbrace \frac{\sigma^2_{tvr}}{\left(1-\bar{\rho}_{tvr}\right)^2}, \bar{\rho}_{tvr}\right\rbrace,
\label{eq:loadbound}
\end{align}
where $\sigma^2_{tvr} = \mathbb{E}[\rho^2_{tvr}] - \bar{\rho}_{tvr}^2$, is the variance of the load, which can also be calculated, similar to $\bar{\rho}_{tvr}$ by using the SINR coverage probability of the typical user.
\label{lem:overl_bound}
\end{lemma}
\begin{proof}
We have for every $k>0$, 
\begin{align}
\mathbb{P}\left[\left(\rho_{tvr} - \bar{\rho}_{tvr}\geq k\sigma_{tvr}\right) \right] \leq \mathbb{P}\left[\left|\rho_{tvr} - \bar{\rho}_{tvr}\right|\geq k\sigma_{tvr} \right]\stackrel{(a)}{\leq} \frac{1}{k^2}, \nonumber
\end{align}
where, (a) is from Chebyshev inequality. Substituting $k \cdot \sigma_{tvr} = 1- \bar{\rho}_{tvr}$, we obtain the first term of the right hand side in  \eqref{eq:loadbound}.
The second term is a direct result of Markov inequality.
\end{proof}
\subsection{Average User Throughput}
The average downlink throughput that a user receives from a BS of type $tvr$ is 
$T_{tvr} \overset{\Delta}{=} \frac{wA_{tvr}}{N_{tvr}}$,
where $N_{tvr}$ is the average number of active users in a cell, which can be approximated by using the mean cell approach \cite{blaszczyszyn2014user}.
The mean cell is defined as a hypothetical cell that has the same average load as that of a typical cell. 

\begin{lemma} The downlink average user's throughput in a non-overloaded mean cell of type $tvr$ is:
$$
T_{tvr} = \lambda \cdot \sigma \frac{1-\bar{\rho}_{tvr}}{\bar{\rho}_{tvr}} {A}_{tvr}.
$$
\end{lemma}
\begin{proof}
The proof is similar to that presented in \cite{bonald2003wireless}.
\end{proof}

The average user throughput is then given by theorem of total probability as:
$T = \sum_{tvr}\mathbb{P}_{tvr}T_{tvr}$. 
 Due to the different operating bandwidths, the bias values which provide the optimal user throughput may lead to weak SINR, which in turn increases the outage. Thus, to guarantee the communication reliability, it is necessary to consider an SINR constraint on the selection of the optimal biases. We define the outage probability with respect to a SINR threshold $\gamma_{\min}$ as:
\begin{align}
\mathbb{P}_{o,tvr}(\gamma_{\min}) = 1 - \mathbb{P}_{Ctvr}(\gamma_{\min}).
\label{eq:outage}
\end{align}
Therefore, we introduce the notion of effective throughput, which measures the throughput of the users, which are not in outage, as:
$T_{eff}(\gamma_{\min}) = \sum_{tvr}\mathbb{P}_{tvr}\cdot T_{tvr} \cdot \mathbb{P}_{Ctvr}(\gamma_{\min}).$ In Section~\ref{sec:PTD}, we optimize $Q_T$ and $Q_R$ so as to maximize the average effective user throughput $T_{eff}(\gamma_{\min})$ under the constraint of a maximum outage probability $\mathbb{P}_{o,tvr}(\gamma_{\min})\leq \bar{\mathbb{P}}_o$ for every BS type $tvr$.

\subsection{{Delay-Throughput Trade-off of the One-Step Association Scheme.}}

{
It must be noted that the sub-optimality in biased received power does not always deteriorate the downlink user throughput, specially for larger access delays. To illustrate this, let us assume that the initial access using mm-wave suffers from a delay given by $\Delta$. In this regard, the throughput for the users associated to the SBSs in mm-wave RAT is given by:}
{
\begin{align}
T_{Svm} = \frac{\sigma}{\frac{N_{Svm}}{\Lambda} + \Delta},
\label{eq:delay_throughput}
\end{align}
{where $\Lambda = \lambda \cdot A_{Svm}$ is the traffic arrival rate in terms of users per second in the mm-wave cell of visibility state $v$ of coverage area $A_{Svm} = \frac{\mathbb{P}_{Svm}}{\lambda_S}$,   $N_{Svm} = \frac{\bar{\rho}_{Svm}}{1 - \bar{\rho}_{Svm}}$ is the number of active users in the cell, and $\frac{N_{Svm}}{\Lambda}$ is the transmission time according to Little's theorem~\cite{leon2008probability}.}
In Section VI, we provide some numerical results to show that in case of realistic access delay with the mm-wave RAT, our scheme performs better in terms of the downlink throughput.}

	\section{Simulation Results}
    \label{sec:PTD}
In this section, we first validate our path-loss exponent approximation with respect to 3GPP values. Then, we study the effects of biases on SINR and user throughput. Finally, we discuss the selection of optimal biases. 
\subsection{Validation of the Path-loss Exponent Approximation}
Fig. 1 shows the comparison of our analytical results {using the approximated path-loss exponents from Table I}  (see Eq. 9) with Monte-Carlo simulations {with actual path-loss exponents from the 3GPP recommendations~\cite{36.814,38.900}} in terms of SINR coverage probability for various tier, RAT selection biases, and density values. Our results indicate that the analytical expressions based on approximated path-loss exponents provide good approximations to the simulated results with 3GPP values of exponents. {Hence, this approximation} can be used for analyzing the system performance. 

\begin{figure} 
\centering
\includegraphics[width=8cm,height = 4.5cm]{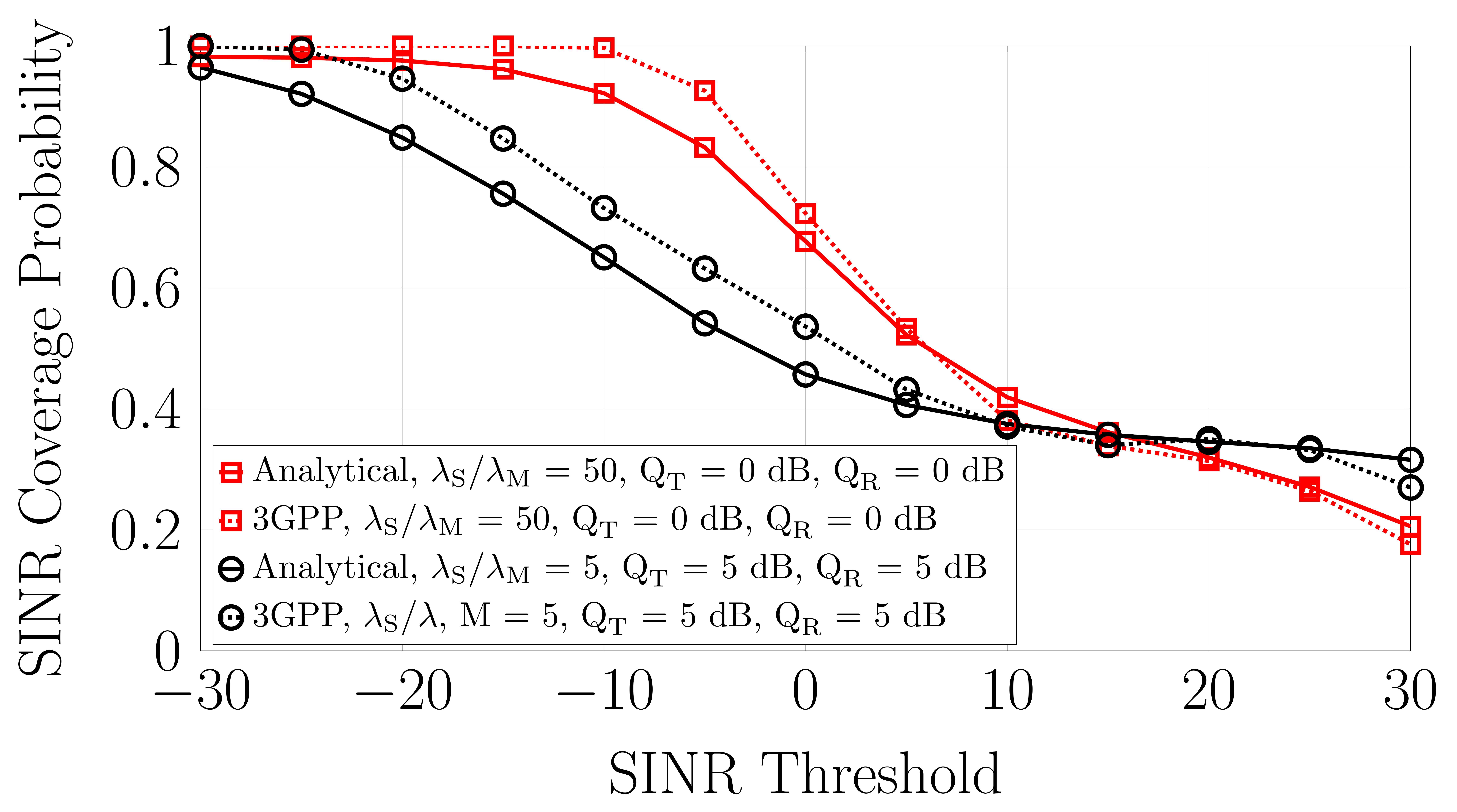}
\caption{Validation of the approximated path-loss exponents with 3GPP parameters.}
\label{fig:AnaSim}
\end{figure}
\subsection{Trends in Cell Association Probabilities}
\begin{figure*}[!t]
\centering
{\includegraphics[width=8cm,height = 4.5 cm]{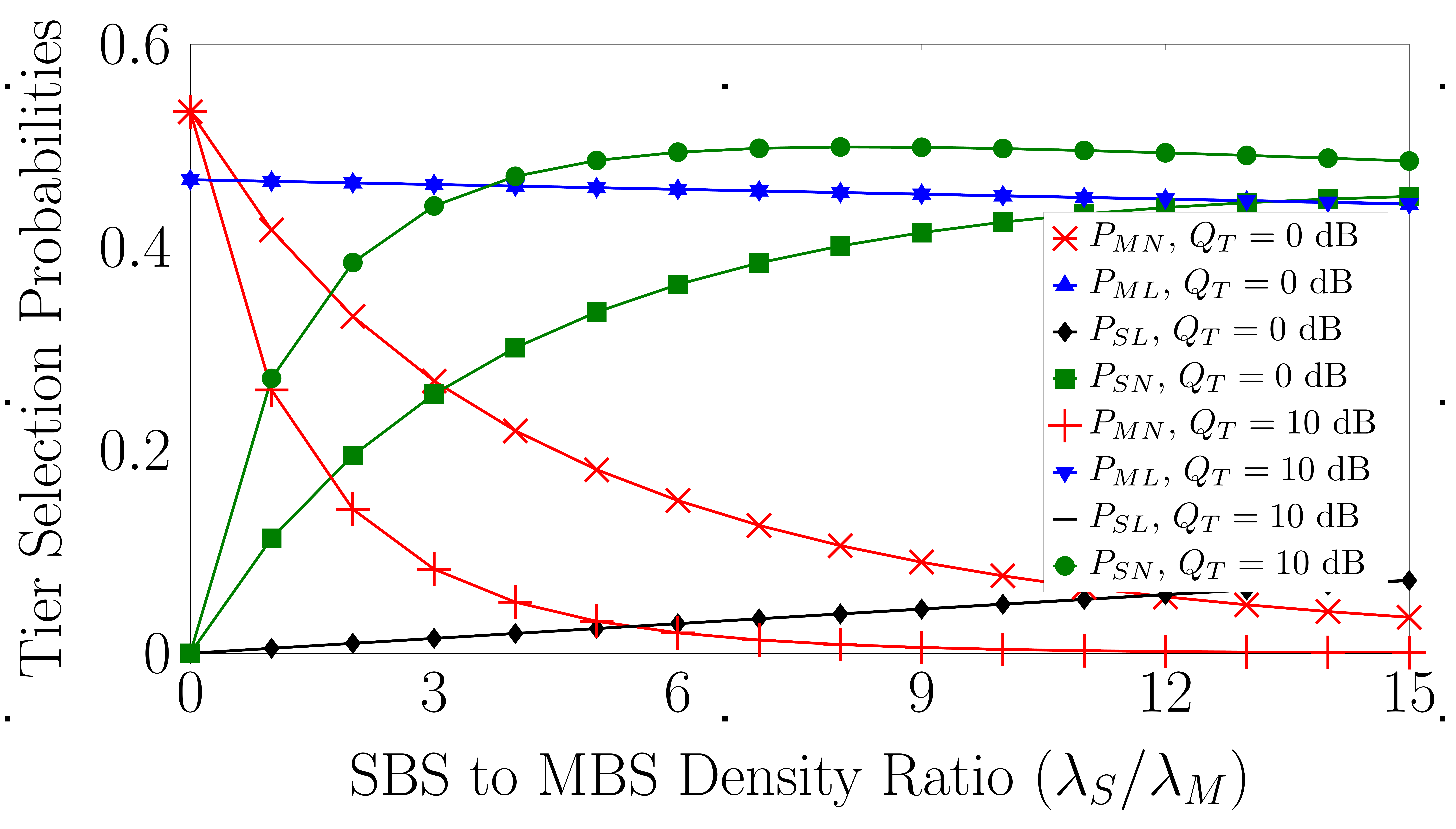}
\label{fig:TierSelection}}
\hfil
{\includegraphics[width=8cm,height =4.5cm]{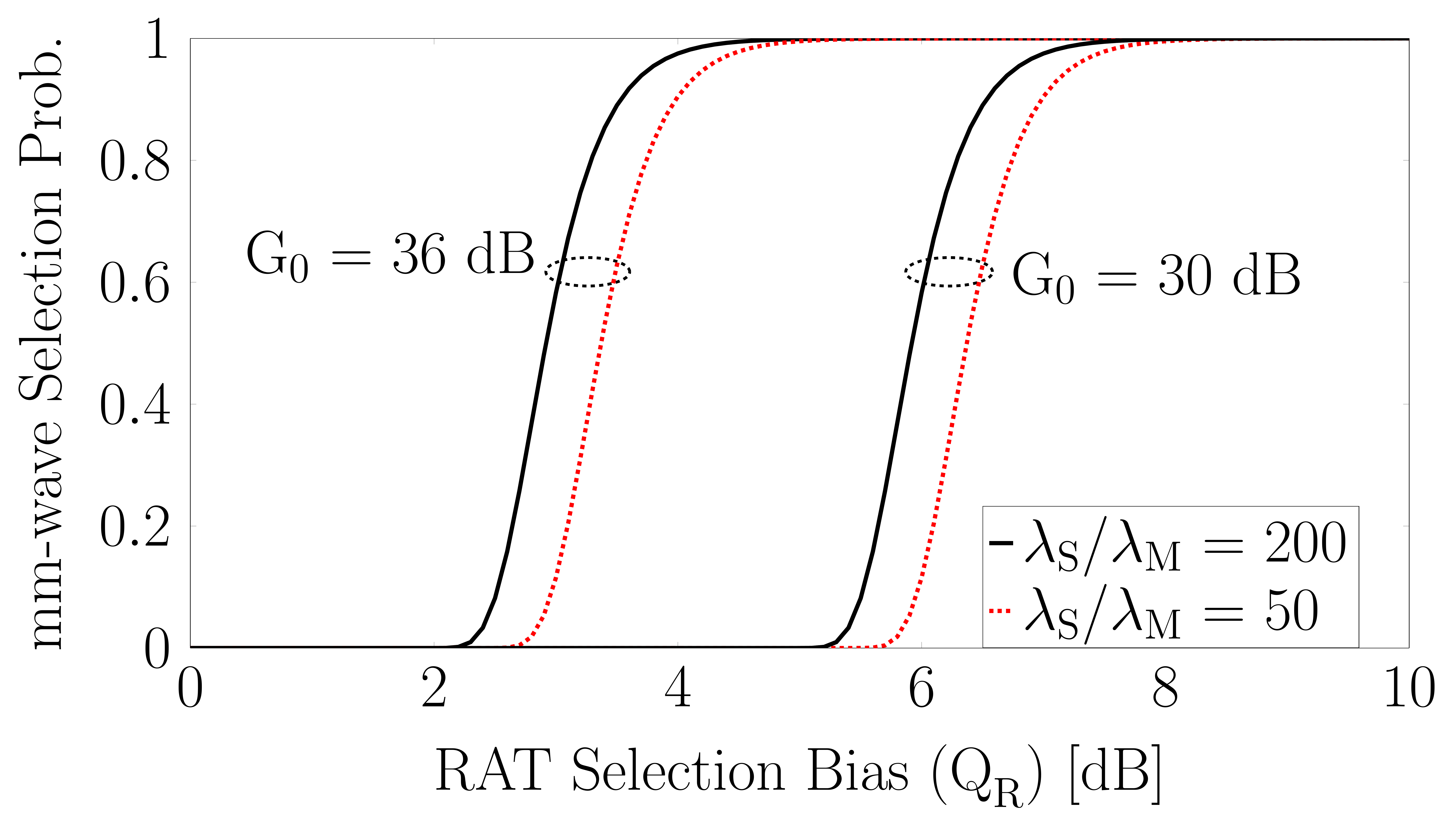}
\label{fig:RATSelection}}
\caption{left) Tier selection probability; right) Conditional mm-wave RAT selection probability with 3GPP parameters.}
\label{fig:ASSOP}
\end{figure*}
Fig. \ref{fig:ASSOP} (left) shows the tier selection probabilities with respect to the ratio of the MBS and SBS densities $\lambda_S/\lambda_M$ with $Q_T = 10$ dB and with $Q_T = 0$ dB. As expected, the association to MBSs decreases as $\lambda_S/\lambda_M$ increases or when $Q_T$ increases. However, the association to LOS BSs does not change appreciably when increasing $Q_T$ from 0 to $10$~dB. Only cell edge users, which are more likely to be in NLOS visibility, are indeed affected by moderate values of $Q_T$. 


The conditional probability of mm-wave service, given that the user has associated with a SBS, is plotted in Fig. \ref{fig:ASSOP} (right), by varying $Q_R$ for two different antenna gains and deployment density ratios. As expected, this probability increases with $Q_R$. However, it is interesting to note that the maximum directional antenna gain has a large effect on the RAT selection regardless of the SBS density. For example, increasing by only 3~dB the antenna gains of transmitter and receiver each has much more impact on the mm-wave association than deploying four times more SBSs. 


\begin{figure*}[!t]
\centering
{\includegraphics[width=8cm,height =4.5cm]{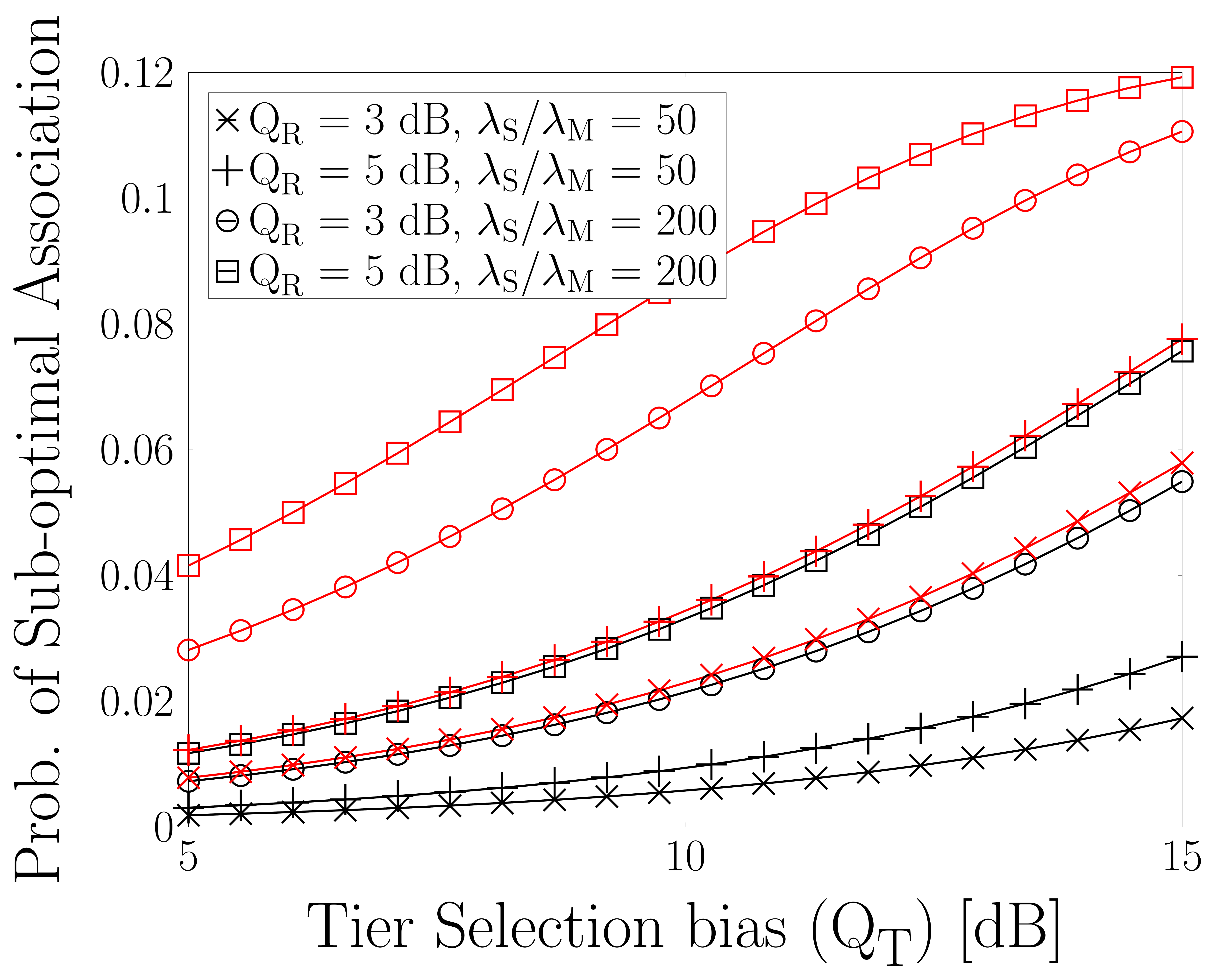}
\hfill
\includegraphics[width = 8 cm , height = 4.5 cm]{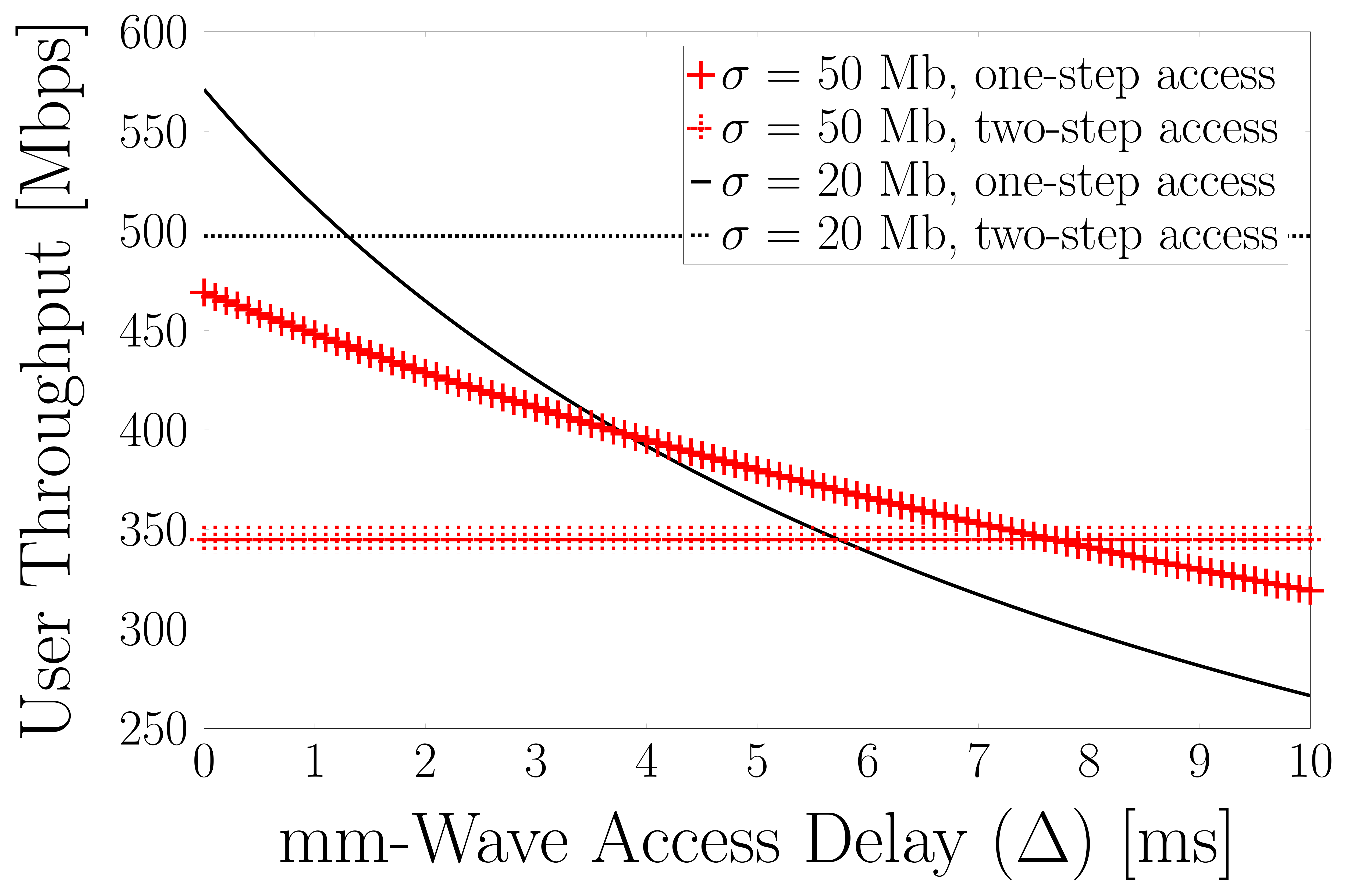}
}
\caption{left) Probability of sub-optimal association for different $Q_R$ and for different $\frac{\lambda_S}{\lambda_M}$. The black curves correspond to $G_0 = 30$ dB and the red curves correspond to $G_0 = 36$ dB; right) Delay-throughput trade-off with mm-wave initial access; $G_0 = 36$  dB, $Q_T = 10$ dB, $Q_R = 5$ dB, $\lambda = 100$ [user $\cdot$ km$^{-2}$ s$^{-1}$].}
\label{Fig:Sub_Opt}
\end{figure*}
\subsection{{Comparison {with the One-Step Association Strategy:}}}
{
We plot the probability of sub-optimal association \eqref{eq:sub_opt} in the left side of Figure~\ref{Fig:Sub_Opt}, for various tier and RAT selection biases and two antenna gains. We note that the probability of sub-optimal association is low $(\leq 12\%)$. Moreover, the probability {becomes} negligible with low tier selection bias $(\leq 1 \%)$. This is because with lower $Q_T$, the biased received power of the mm-wave transmission in SBS are lower, thereby reducing the probability of sub-optimality. Similarly, with lower antenna gain $(G_0)$, the biased mm-wave power is lower, resulting in low sub-optimality. Furthermore, we observe that the probability of sub-optimal association increases with increasing network densification, since denser networks correspond to higher mm-wave powers. However, for $G_0 = 30$ dB, the probability of sub-optimal association does not exceed $8 \%$ even for very dense deployments. }

{
In the right side of Figure~\ref{Fig:Sub_Opt}, we compare the throughput perceived by the typical user with the two approaches \eqref{eq:delay_throughput}. We plot the downlink user throughput vs the initial access delay $\Delta$, for two different file sizes $(\sigma)$. We see that with increasing $\Delta$, the throughput with the one-step association scheme decreases, and goes below the throughput achieved by using our two-step solution. In practical systems, the initial access delay in mm-wave can be of the order of several milliseconds~\cite{li2016initial}. As a result, our two step association is more efficient in terms of the user throughput as compared to the case where association is performed in one-step. }  

\subsection{Trends in SINR Coverage Probabilities}
In Fig. \ref{fig:CovPQ1}, we plot the SINR coverage probability of the typical user, with respect to $Q_T$ and various ratios of SBS to MBS densities.} In the case where the SBSs operate only in sub-6GHz band, i.e., $Q_R = -\infty$~dB (Fig. \ref{fig:CovPQ1} left), increasing the tier selection bias decreases the SINR coverage probability because some users are forced to associate with BSs providing less signal power. For $d_M=200$~m and $d_S=20$~m, Fig. \ref{fig:CovPQ1} (right) shows that the same effect can be observed when SBSs transmit data only in mm-wave, i.e., $Q_R = \infty$, regardless of the deployment densities. 

However, in the latter case, when varying the blockage characteristics, we observe two different behaviors for the SINR coverage probability. In Fig.~\ref{fig:CovP2}, we see that depending on the LOS ball radii, the SINR may increase with the tier selection bias. Indeed, the SINR may improve by associating macro cell users to SBSs transmitting data only in mm-wave, even though this SBS offers less power in sub-6GHz band, because the received power in mm-wave may be higher due to antenna gain. Additionally, the interference in mm-wave is generally lower than the one in the sub-6GHz band. However, increasing the bias further forces the users closer to the MBS to associate with a SBS that provide very limited received power, which leads to lower SINR. 
\begin{figure*}[!t]
\centering
{\includegraphics[width=8cm,height =4.5cm]{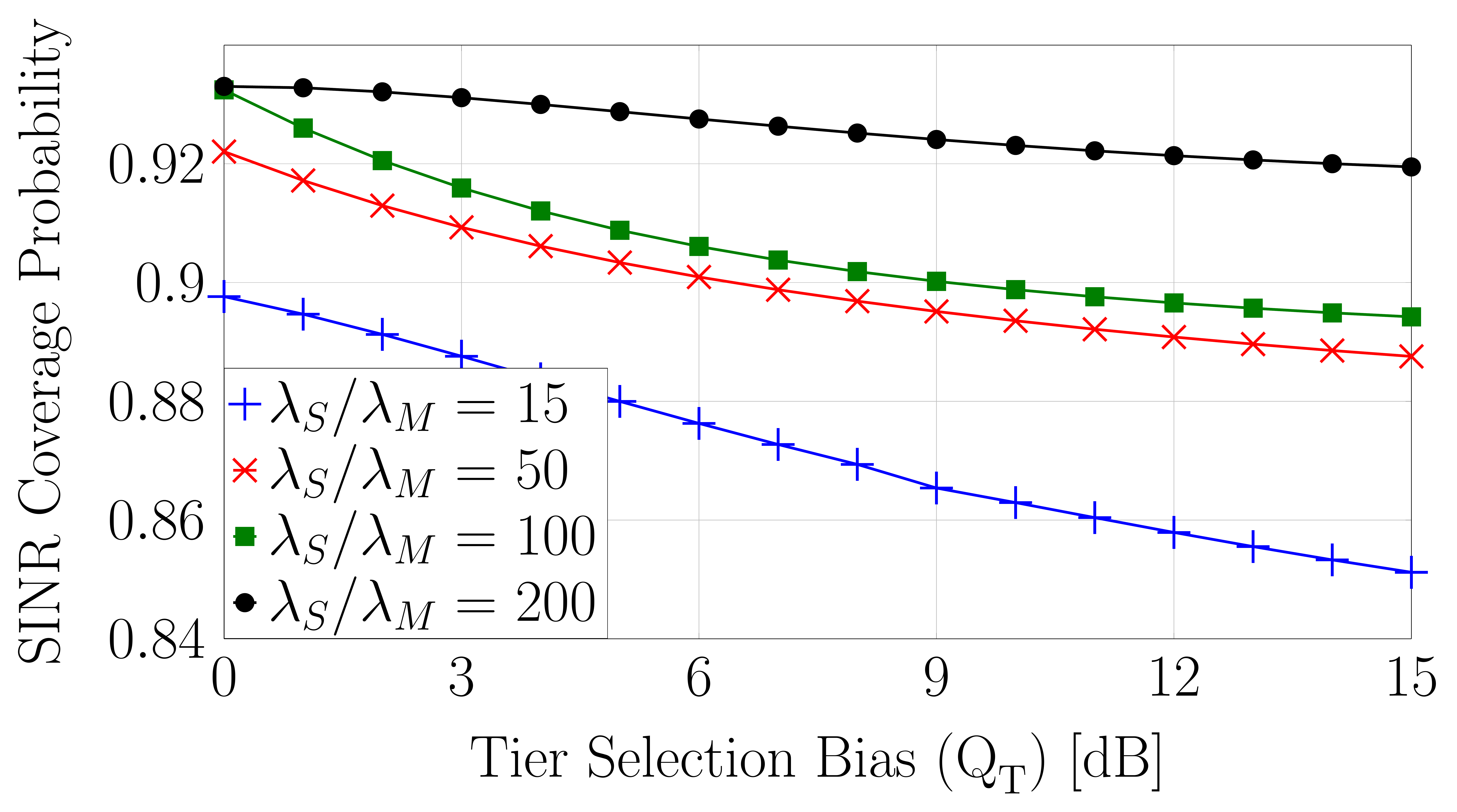}
\label{fig:CovPmu}}
\hfil
{\includegraphics[width=8cm,height =4.5cm]{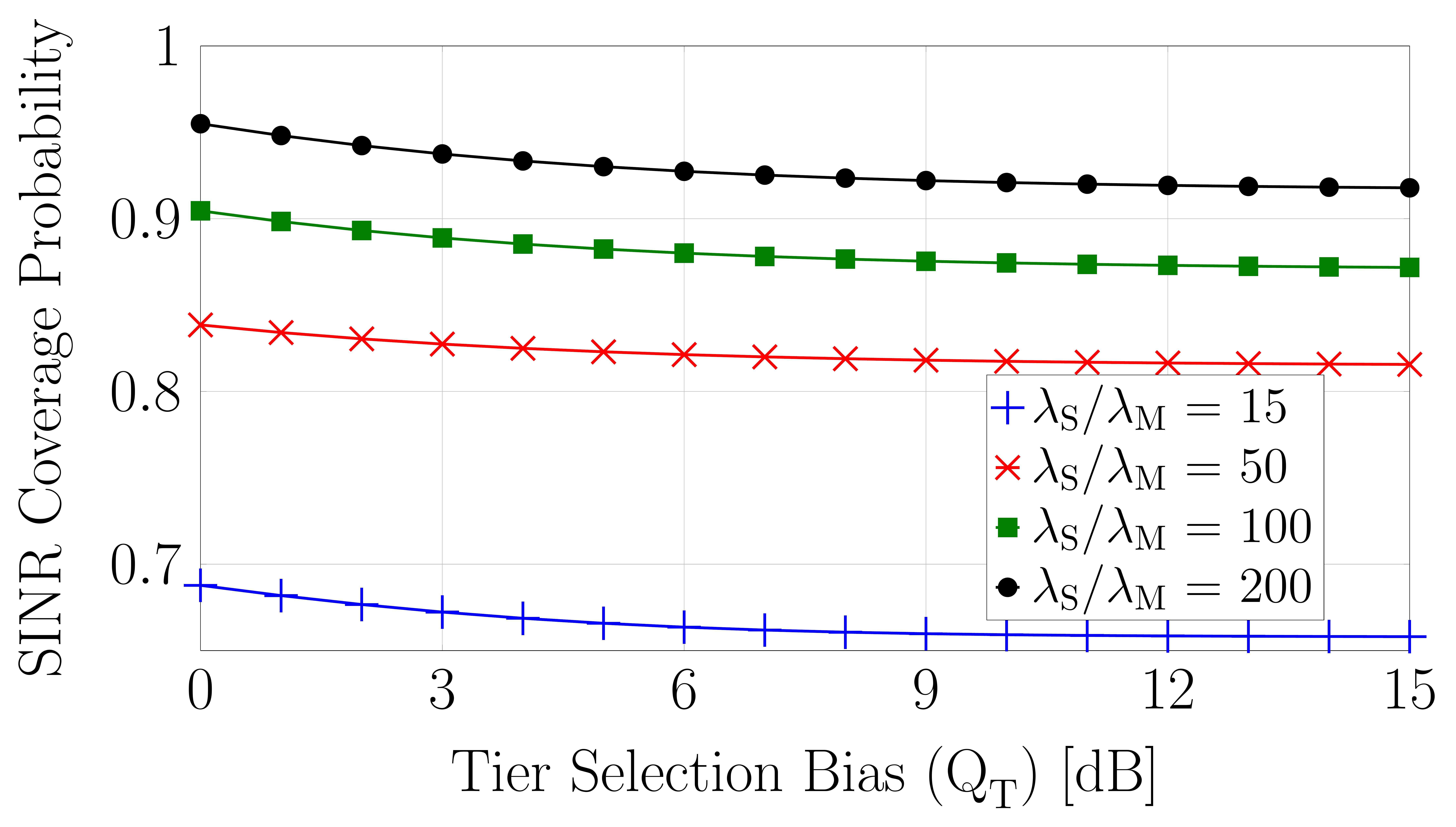}
\label{fig:CovPmm}}
\caption{SINR coverage probability vs tier selection bias at a threshold of $\gamma = -10$ dB for left) $Q_R=-\infty$~dB; right) $Q_R=\infty$, $d_M=200$~m, $d_S=20$~m.}
\label{fig:CovPQ1}
\end{figure*}
\begin{figure*}[!t]
\centering
\includegraphics[width=8cm,height = 4.5cm]{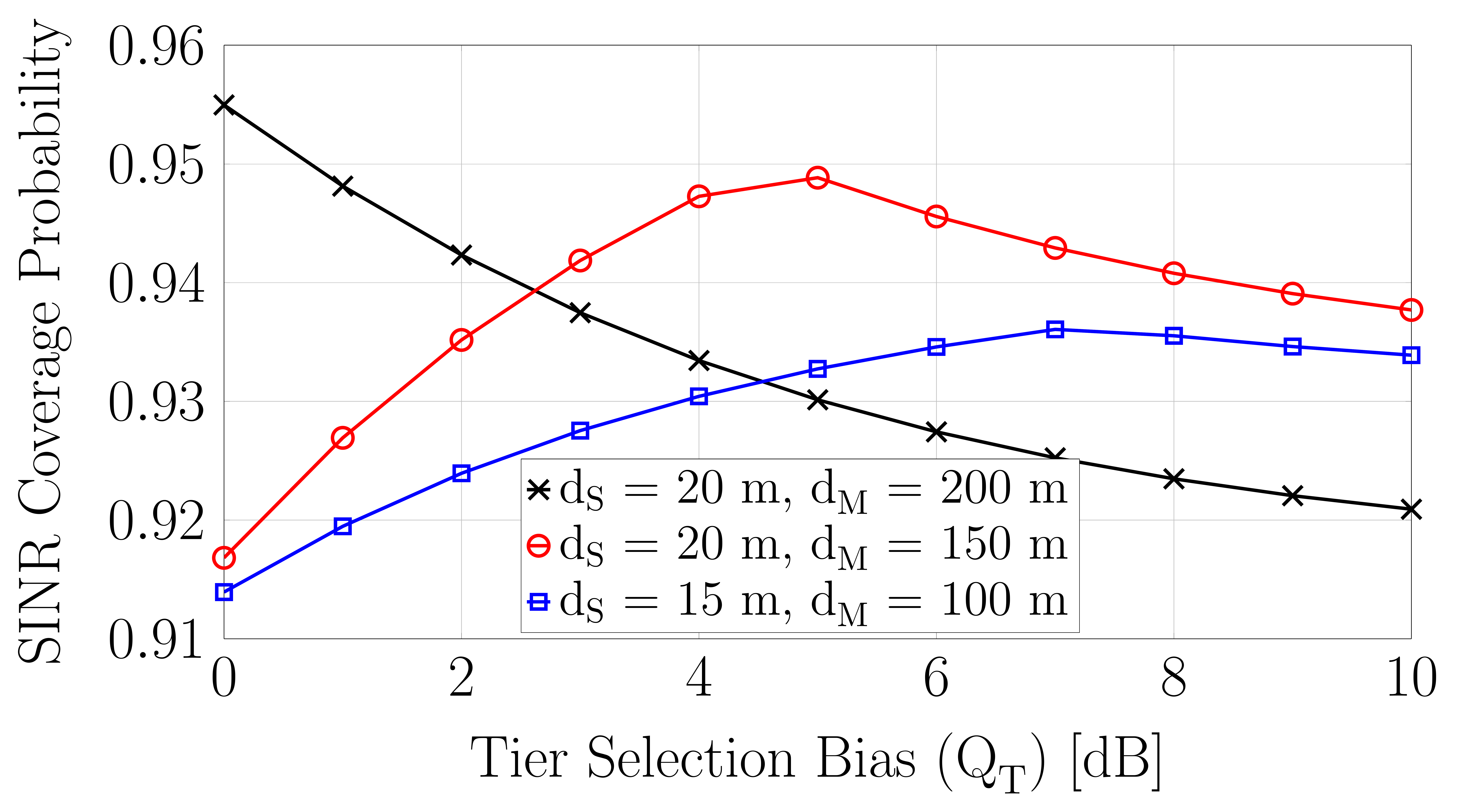}
\hfil
\includegraphics[width=8cm,height =4.5cm]{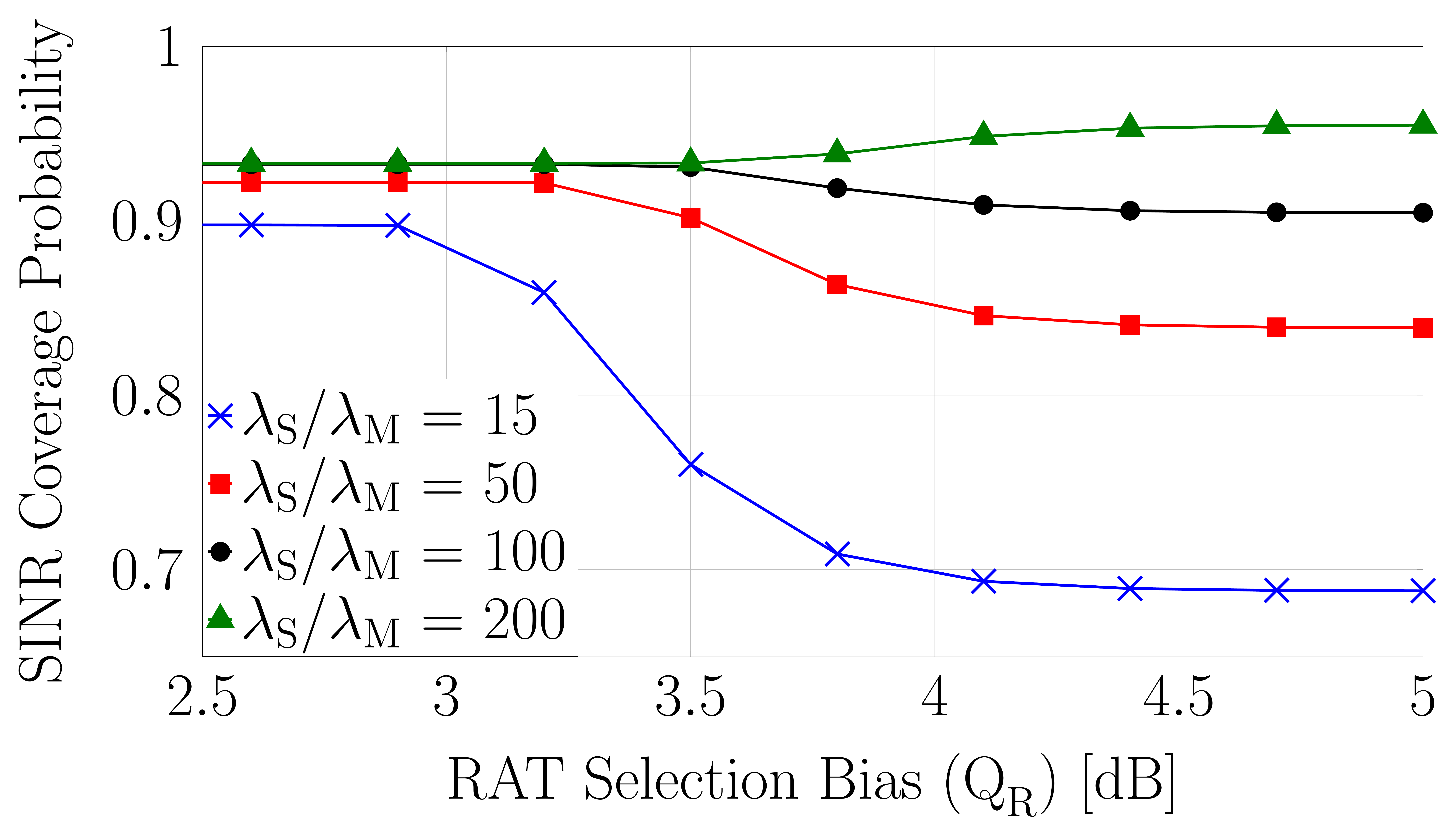}
\caption{left) SINR coverage probability with different LOS radii for $Q_R=\infty$, $\lambda_S/\lambda_M = 200$; right) SINR coverage probability vs RAT selection bias with $Q_T = 0$ dB.}
\label{fig:CovP2}
\end{figure*}

Assuming maximum power tier selection ($Q_T=0$~dB), Fig.~\ref{fig:CovP2} (right) shows that increasing $Q_R$ has contrasting effects on the SINR depending on the ratio of SBS to MBS densities. Increasing the SBS density increases co-channel interference more in sub-6GHz than in mm-wave. Moreover, as the user-SBS distance decreases, the useful signal power increases more in mm-wave than in sub-6GHz. Both effects are due to the difference in the path-loss models. As a consequence, as the SBS density increases ($\lambda_S/\lambda_M = 200$), it is more and more attractive for a user to be served by mm-wave band, which is realized by higher values of $Q_R$. On the contrary, in case of sparser SBS deployments ($\lambda_S/\lambda_M = 15, 50, 100$), increasing $Q_R$ forces users to be served from distant SBSs in mm-wave, and the gain due to the reduced interference cannot compensate the signal strength loss. Note that this contrasting effect cannot be observed with single RAT networks.  

We now study the joint effect of $Q_T$ and $Q_R$ for dense ($\lambda_S/\lambda_M = 200$) and sparse ($\lambda_S/\lambda_M = 50$) deployments in Fig. \ref{fig_PCS} left and right, respectively. For sparse SBS deployments, the conclusions drawn so far hold: high SINR regions occur at low $Q_T$ and $Q_R$. The optimum biases as marked in the figure are $Q_T = 0$ dB and $Q_R = 0$ dB. For dense deployments, however, we can observe that, for a high $Q_T$ (here for $Q_T>8$~dB), SINR coverage probability generally decreases with increasing $Q_R$, which is in contrast to the case when $Q_T$ is small. This is because, for users far away from their serving SBS, it is now preferable to get associated with sub-6GHz than with mm-wave. The optimal biases in this case are $Q_T = 0$ dB and $Q_R = 5$ dB.  


\begin{figure*}[!t]
\centering
{\includegraphics[width=8cm,height = 4.5cm]{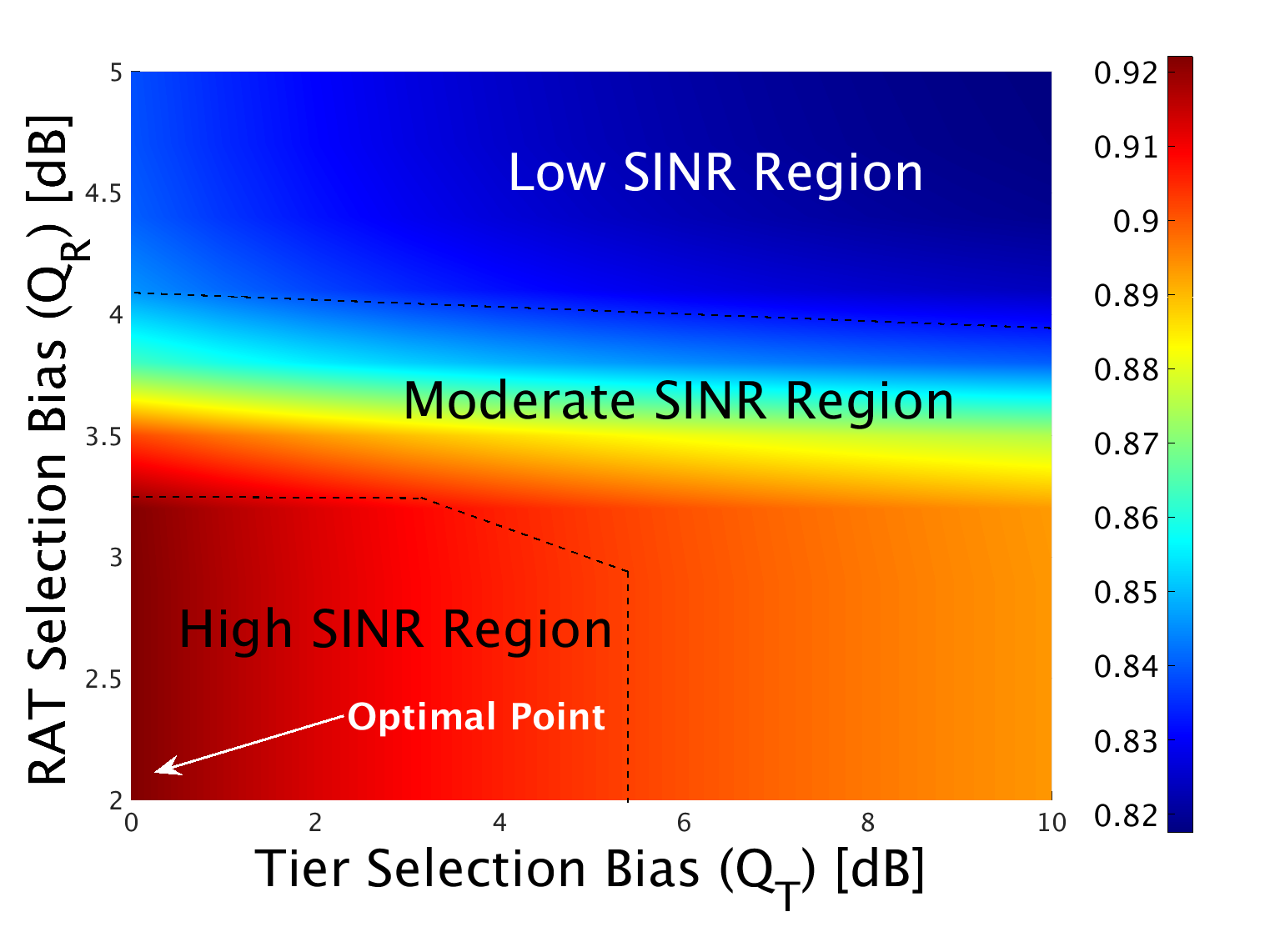}
\label{fig_PCS15}}
\hfil
{\includegraphics[width=8cm,height = 4.5cm]{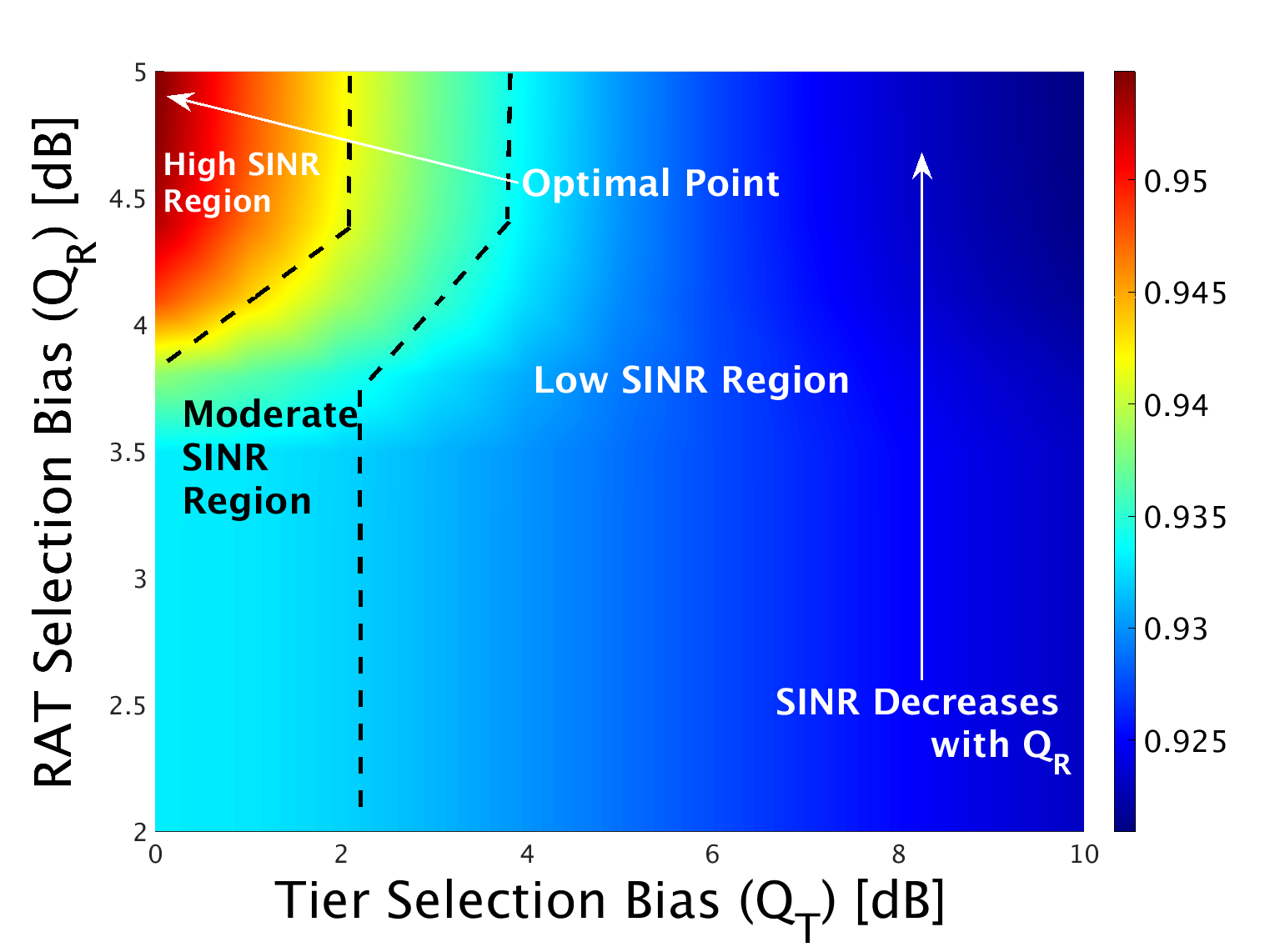}
\label{fig_PCS5}}
\caption{SINR coverage probability as a function of $Q_T$ and $Q_R$ at a threshold of $\gamma = -10$ dB for $d_M=200$~m, $d_S=20$~m for left) $\lambda_S/\lambda_M = 50$ right) $\lambda_S/\lambda_M = 200$.}
\label{fig_PCS}
\end{figure*}
\subsection{Performance of the Near-Optimal Strategy to Select Tier and RAT Biases}
{
 In Section \ref{sec:bias selection}, we have proposed a near-optimal strategy to fix the RAT and tier-selection bias, to reduce the complexity of the brute force search. In this strategy, first $Q_R$ is selected according to \eqref{eq:SINRratio}. Then, for the fixed $Q_R$, a $Q_T$ is selected according a random-restart hill-climbing algorithm as described in Algorithm 2.
 We show the convergence of the algorithm in the left side of Fig. 7 for $\lambda_S/\lambda_M = 200$ and for two pairs of LOS radii. With $k = 0.5$ and $\beta = 2$, the proposed algorithm converges at $Q_T = 3.19$ dB for $d_S = 20$~m and $d_M = 150$~m, and at  $Q_T = 7.16$ dB for $d_S = 15$~m and $d_M = 100$~m.
Fig. 7 (center) compares various bias selection strategies. We observe that our proposed strategy performs at least as good as the classical strategy based on the maximum received power. In particular, for sparse deployment of SBSs, the proposed strategy and the maximum power association perform equal to the optimal association. However, with increasing SBS density, the performance of the maximum power association decreases due to the increasing interference in the sub-6GHz RAT. On the contrary, since our strategy takes interference into account, it achieves near optimal SINR.}

\begin{figure*}[!t]
\centering
\includegraphics[width=5cm,height = 4.5cm]{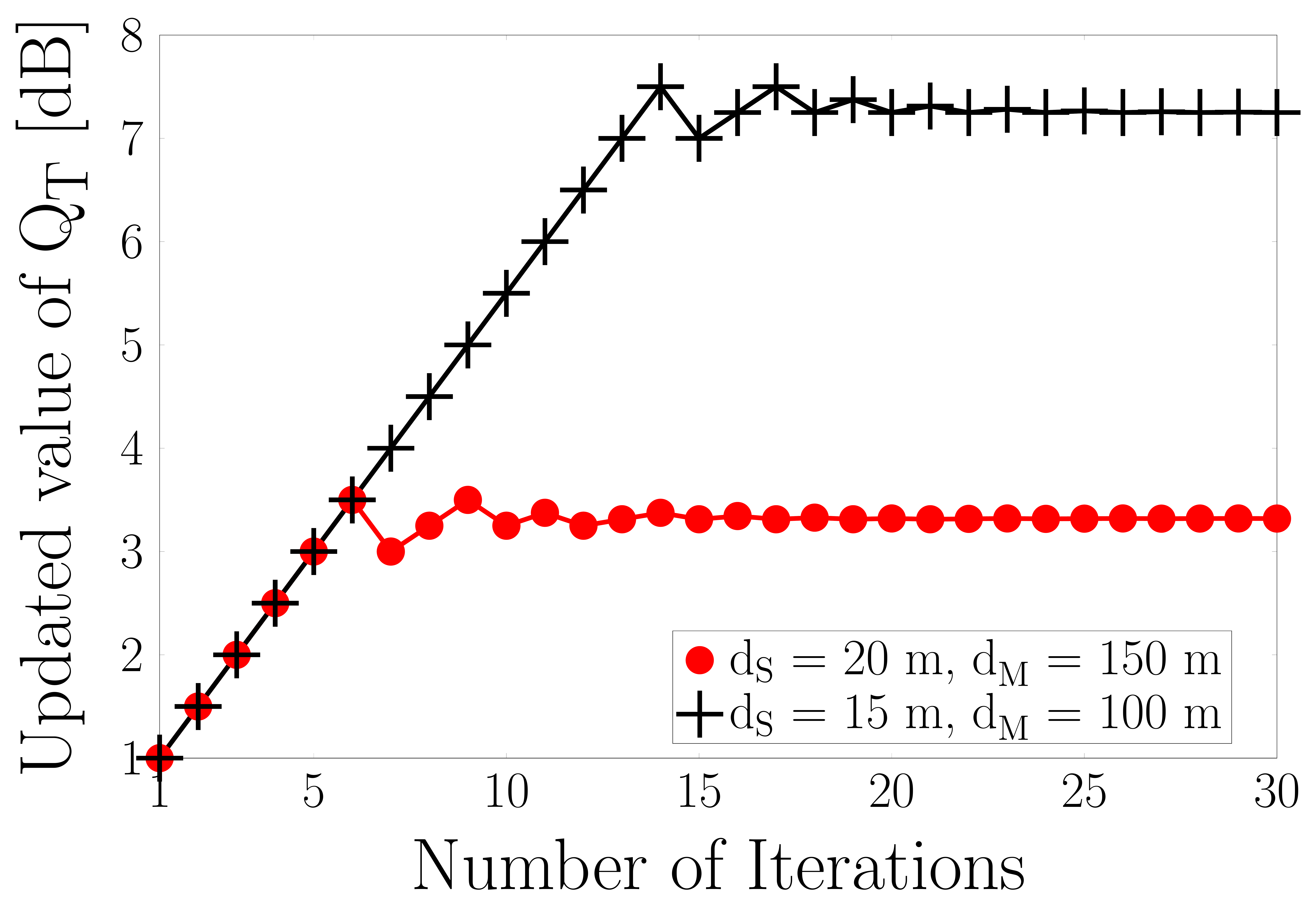}
\includegraphics[width=5cm,height = 4.5cm]{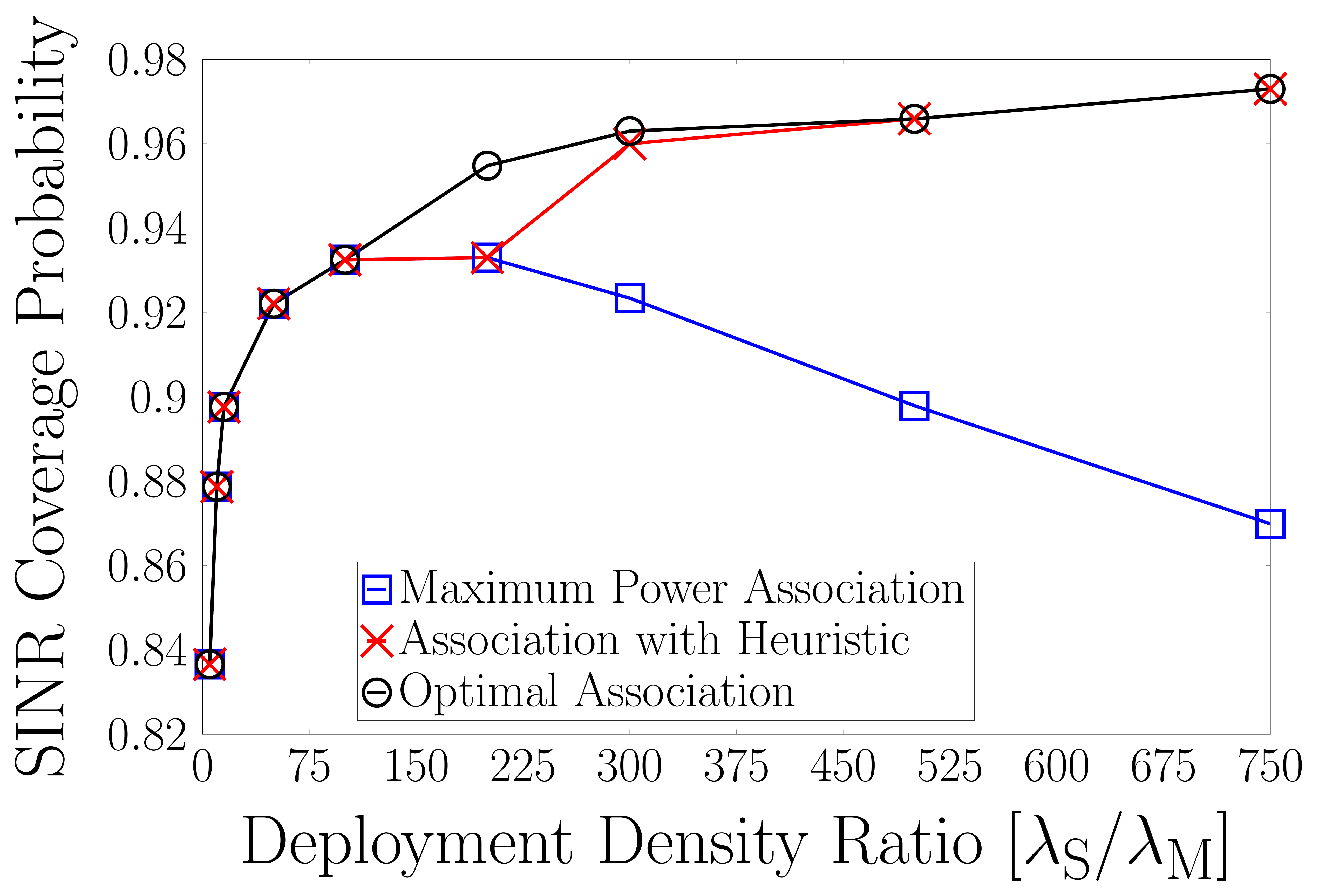}
\includegraphics[width=5cm,height = 4.5cm]{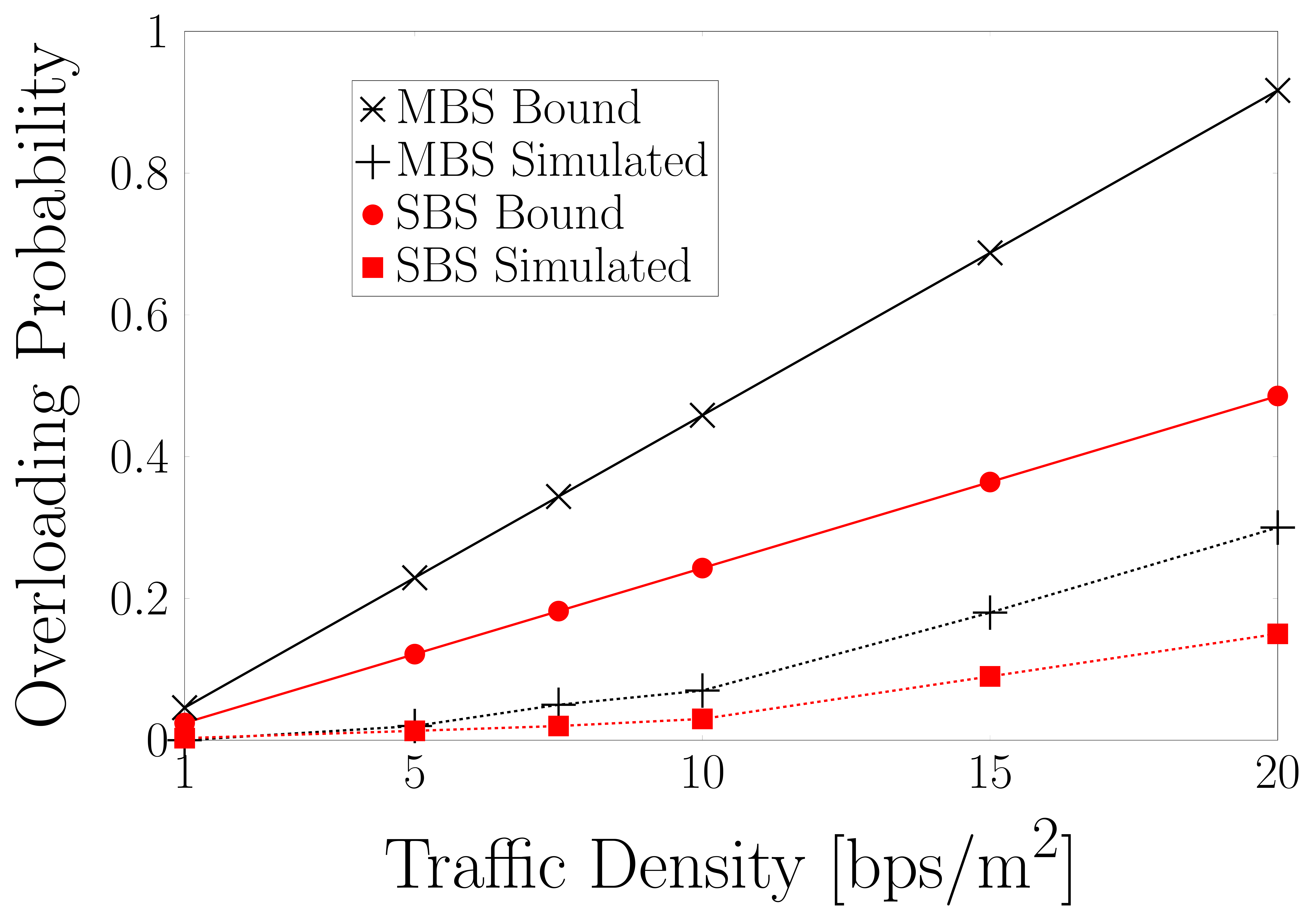}
\caption{left) Convergence of gradient descent algorithm for $\lambda_S/\lambda_M = 100, d_S = 10$ m and $d_M = 100$ m; center) Comparison of RAT selection strategies; right) Tightness of the bound on probability of overloading.}
\label{fig_RAT}
\end{figure*}
\subsection{Analysis of the Bound on Overloading Probabilities} 
In this section, we investigate the relation between the cell overloading probabilities and the traffic density, based on the analytical bound derived in Lemma \ref{lem:overl_bound}. We see in Fig. \ref{fig_RAT} (right) for $\lambda_S/\lambda_M = 5$ that the proposed bound is relatively loose but it provides the operator the guarantee that the overload probability will not exceed this value. Based on this bound and a constraint on the overall outage probability, a conservative network sizing can be derived. In Fig.~\ref{Overload} (left), we show the minimum deployment density required such that feasible biases exist to meet both theses constraints. The more stringent the constraints are, the more SBSs the operator should deploy. When the traffic density is low, the outage probability is the limiting constraint and accordingly, the minimum deployment density is the one required to maintain coverage. However, as traffic density increases, overloading probability is determining. 

\begin{figure*}[!t]
\centering
{\includegraphics[width=6cm,height = 4.5cm]{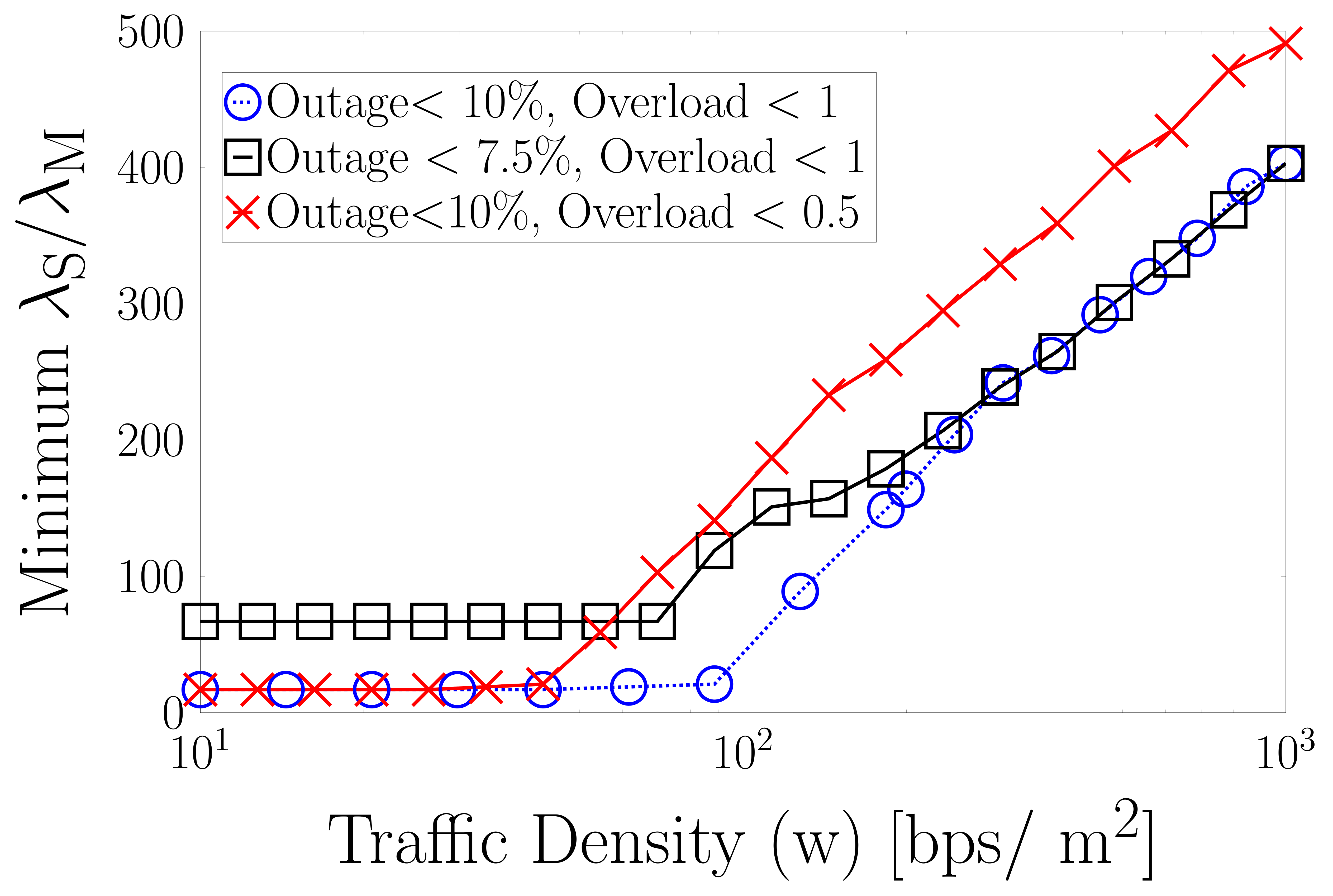}
\label{fig:loadbound}}
\hfil
{\includegraphics[width=10cm,height = 4.5cm]{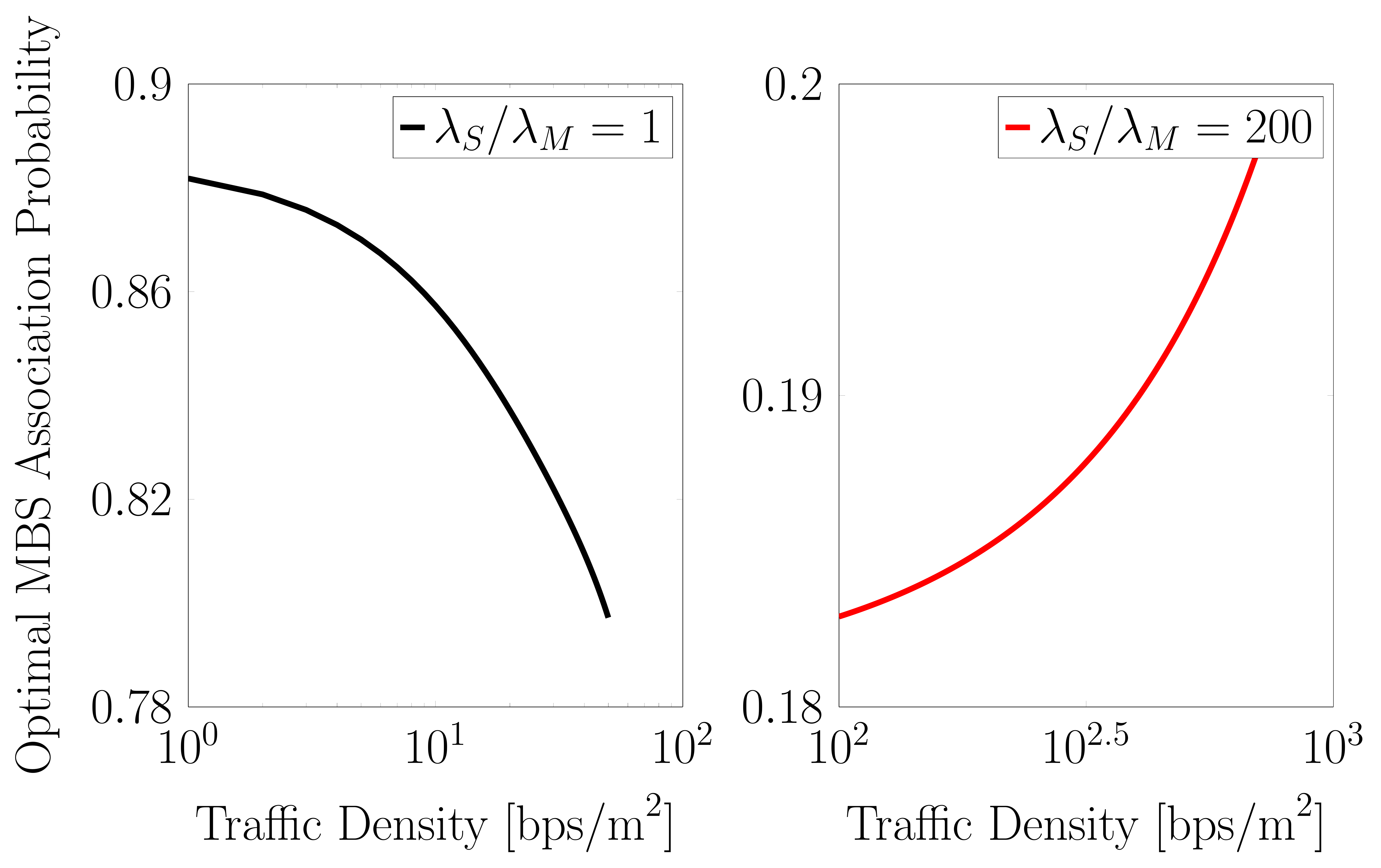}
\label{fig:loadboundSNmu}}
\caption{left) Minimum required deployment density for a given traffic density and right) Throughput optimal MBS association probabilities.}
\label{Overload}
\end{figure*}
\subsection{Rate Optimal Choice of Tier and RAT Selection Biases}
In this section, we optimize tier and RAT selection biases with respect to the average effective throughput. To guarantee a good coverage, we impose a constraint on the outage probability (from 7.5 to 12.5\%\footnote{Note that generally, PPP based modeling of cellular networks provide a pessimistic view of the network. Previous studies showed that an outage of 1\% in hexagonal model corresponds to 10\% outage in a PPP based modeling for the same network parameters~\cite{7488044}.}).     
The MBS association probability corresponding to the optimal $Q_T$ as a function of the traffic density is shown in Fig.~\ref{Overload} (right). Depending on the ratio of densities $\lambda_S/\lambda_M$, users are offloaded from MBS to SBS (for low SBS densities) or vice-versa (for high SBS densities). 


In Fig. \ref{RateetOut} (left), we show the optimal effective downlink throughput as a function of the traffic density for various deployment densities and outage constraints. We observe that more stringent outage constraints result in lower downlink throughput in the network. This is because biases are mainly optimized to guarantee coverage also for cell edge users. We also observe that increasing the SBS density not only results in higher throughput, but also increases the range of traffic densities that the network can serve, i.e., the network capacity. {In this evaluation, we have obtained the downlink throughput by considering that the users in overloaded base stations receive zero throughput. Therefore, even though the network as a whole can serve traffic densities up to 1 Gbps/m$^{2}$, the MBS tier gets overloaded for much lower traffic densities. Accordingly, the network is no longer well-dimensioned for the region of traffic densities beyond the MBS overloading points.} Furthermore, in Fig. \ref{RateetOut} (right), we plot the optimum association probabilities as a function of outage probability with $\lambda_S/\lambda_M = 50$ and traffic density of 200 bits$\cdot$s$^{-1}$m$^{-2}$. We see that for more stringent outage constraints, sub-6GHz service in SBSs becomes necessary, in addition to mm-wave service, to satisfy the QoS constraints of outage and overloading simultaneously, thus justifying the interest of deploying dual band SBSs.
\begin{figure*}[!t]
\centering
{\includegraphics[width=8cm,height = 4.5cm]{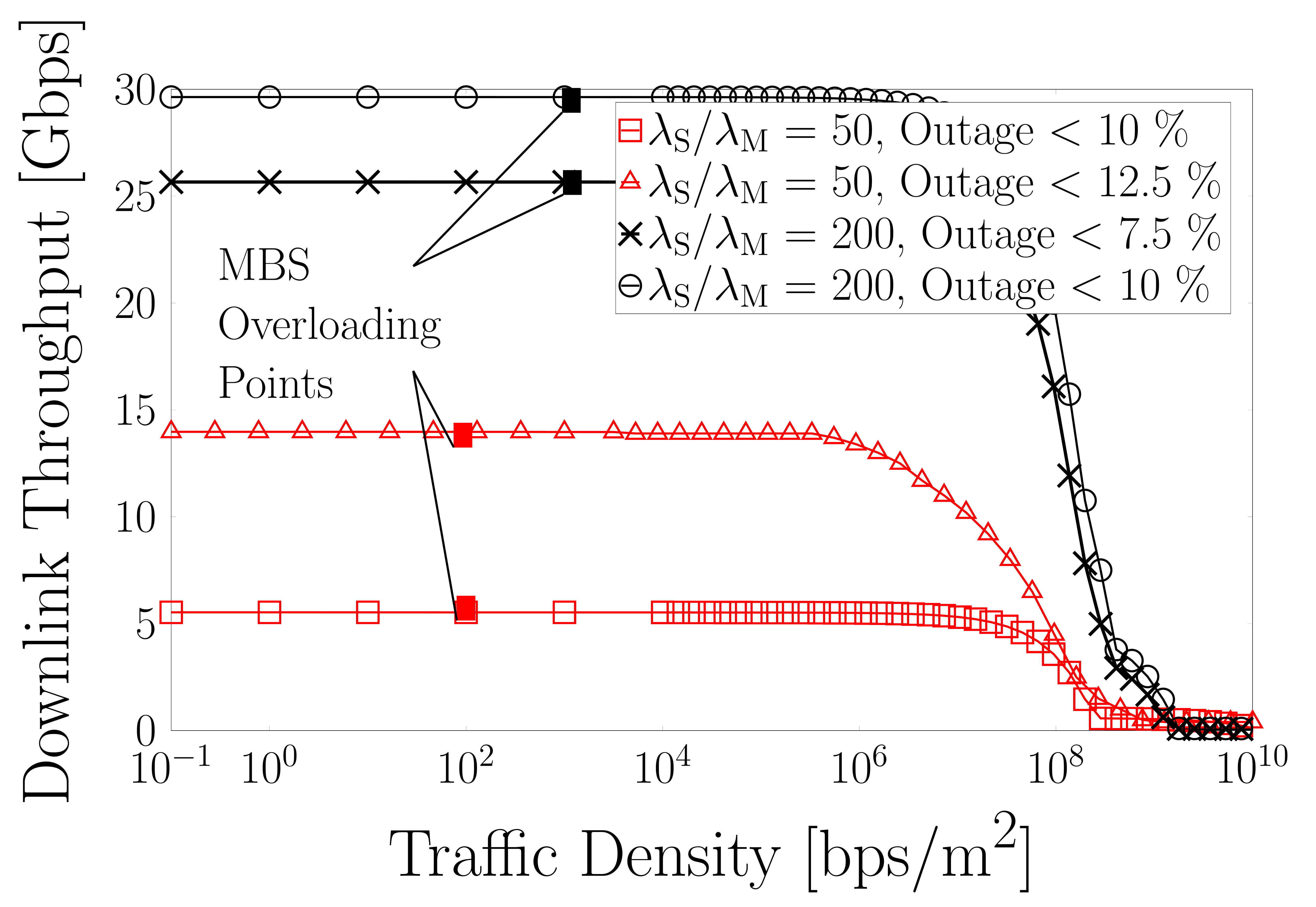}
\label{fig:RATESINR}}
\hfil
{\includegraphics[width=8cm,height = 4.5cm]{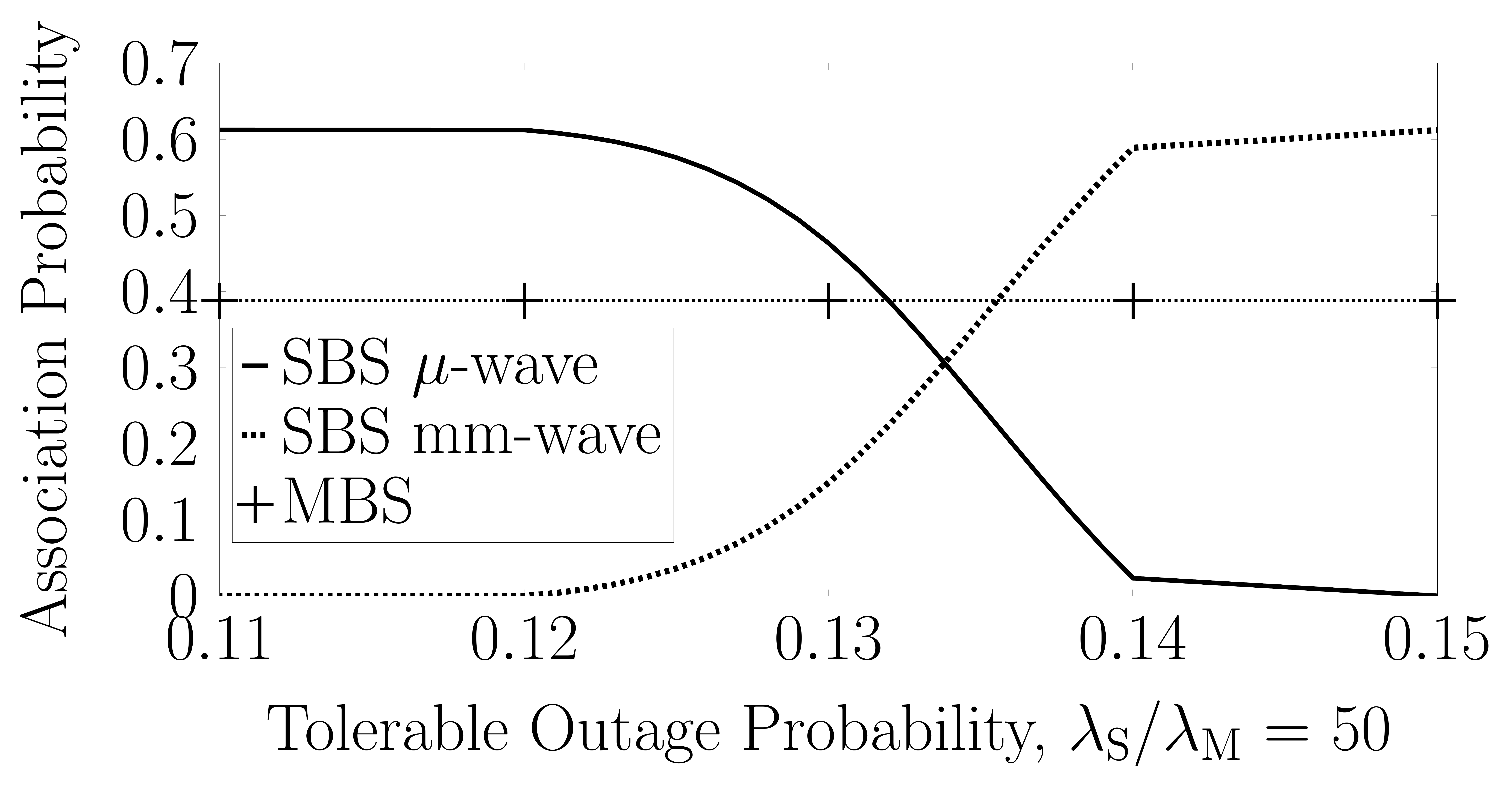}
\label{fig:Outage}}
\caption{left) Optimal downlink user throughput and right) Optimal association probabilities for different outage probability constraints.}
\label{RateetOut}
\end{figure*}
\begin{figure*}[!t]
\centering
{\includegraphics[width=8cm,height = 4.5cm]{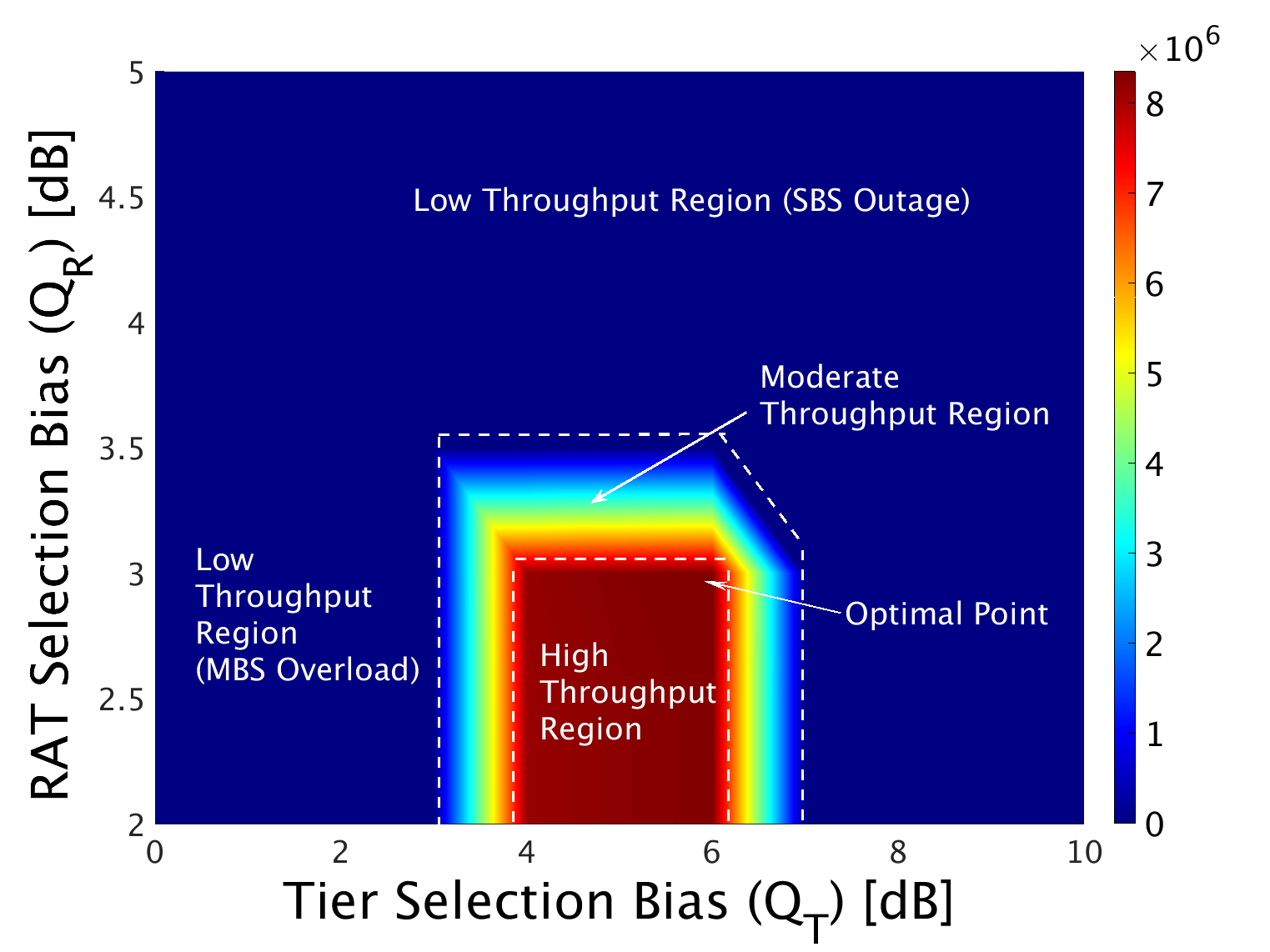}
\label{fig_PCS15}}
\hfil
{\includegraphics[width=8cm, height = 4.5cm]{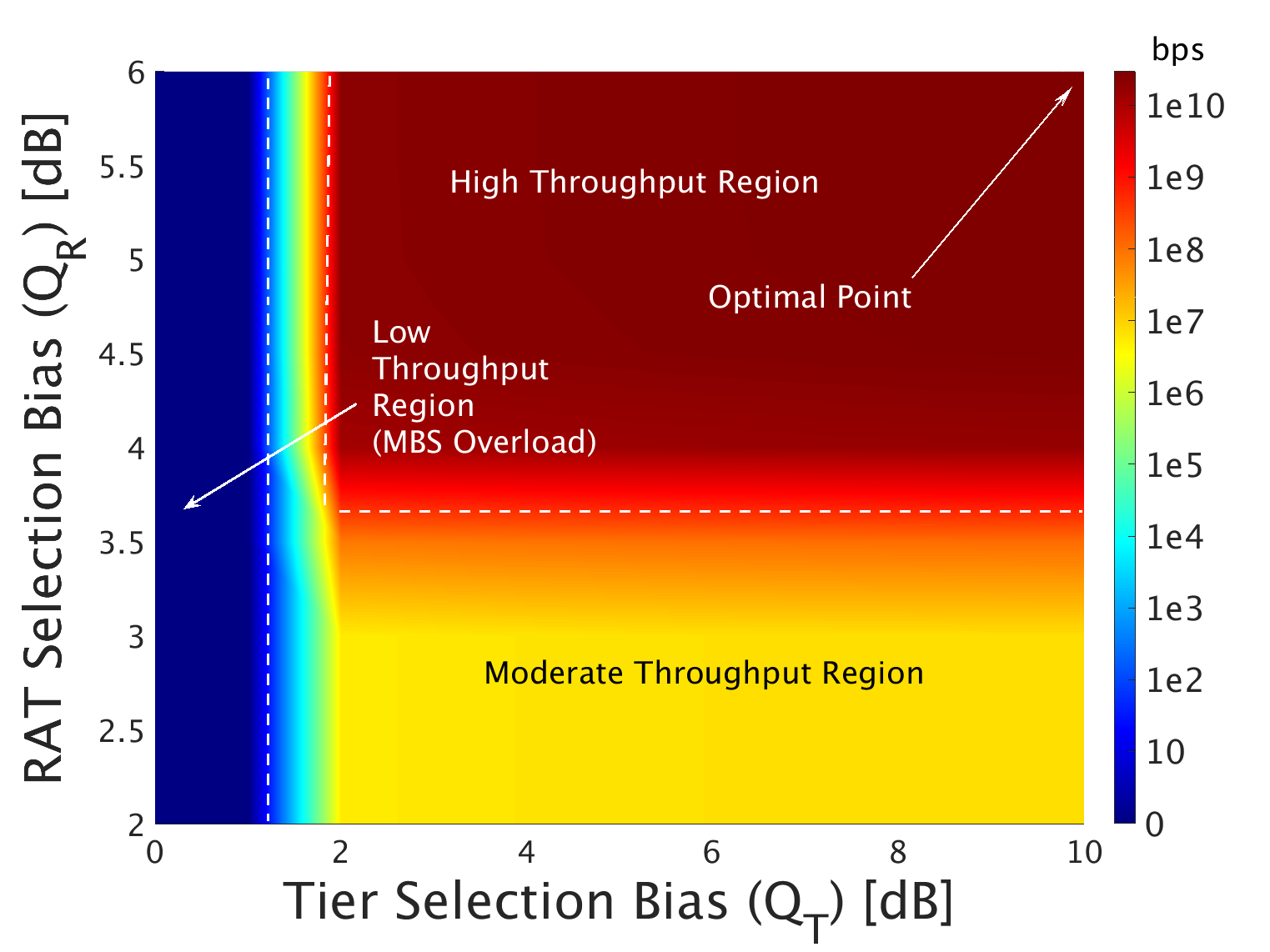}
\label{fig_PCS5}}
\caption{Effective user throughput vs $Q_T$ and $Q_R$ at a traffic density of 100  bits$\cdot$s$^{-1}$m$^{-2}$, tolerable outage probability of 0.10, $d_M=200$~m, $d_S=20$~m for left) $\lambda_S/\lambda_M = 50$ right) $\lambda_S/\lambda_M = 200$.}
\label{fig_RS}
\end{figure*}

As a conclusion, in dense SBS deployments (see Fig. \ref{fig_RS}, right), the users do not suffer from outage even in the case of high tier biases. In this case, $Q_R$ should be high enough to maximize the mm-wave association probability. In case of $\lambda_S/\lambda_M = 200$, this results in a maximum throughput of around 30 Gbps at $Q_T = 10$ dB and $Q_R = 6$ dB. In sparse SBS deployments ( Fig. \ref{fig_RS} (left)), high values of $Q_T$ are desirable to offload traffic from overloaded MBSs. However, as the SBS ranges increase, mm-wave becomes unattractive for users at the SBS cell edges. We can observe that increasing $Q_R$ beyond a certain limit pushes the SBS users in outage thereby decreasing the effective throughput. The maximum average throughput in this scenario, considering the regime of biases where the MBS tier is not overloaded, is 10  Mbps at $Q_T = 6$ dB and $Q_R = 3$ dB.

\section{Conclusion}
{ In this paper, we characterize a two tier network, consisting of classical sub-6GHz macro cells, and Multi RAT small cells, able to operate in sub-6GHz and mm-wave bands. First, we propose a {two-step} tier and RAT selection strategy {where the sub-6GHz band is used to speed-up the initial access procedure in the mm-wave RAT}, and then we investigate the effect of tier and RAT offloading in terms of SINR, cell load, and throughput. Our study highlights the fundamental trade-offs between outage probability, user throughput, and overloading probability, and, thereby, underscores the necessity of the dual band small cells to maintain outage below a certain threshold, specially in sparse deployments. In our system model, we have proposed effective approaches to optimize the user association. However, obtaining closed form solutions for the optimal biases and the maximum traffic density that the network can handle are open challenges. Moreover, the dual band nature of the base stations calls for advanced radio resource management, which is an interesting topic to be investigated.}


	\appendices
    \section{Tier Selection Probability}
	\label{App:TierSel}

The probabilities that at least one LOS MBS and LOS SBS exist are, respectively,
$\mathbb{E}[\mathds{1}(t_ML)]   
= 1 - \exp(-\pi \lambda_M d_M^2)$ and
\label{eq:VOID}
$\mathbb{E}[\mathds{1}(t_SL)]   
=1 - \exp(-\pi \lambda_S d_S^2)$.
Then, the values of $\mathbb{P}(\tilde{Q}_TP_{tv\mu1}>\tilde{Q}_TP_{t'v\mu1})$ are derived as follows:
	\begin{align}\nonumber
	&\mathbb{P}(P_{ML\mu1} > Q_T\cdot  P_{SL\mu1})  = \int_0^\infty e^{-\Lambda'_{SL\mu}(0,Q_Tr)}e^{-\Lambda'_{ML\mu}(0,r)}\lambda'_{ML\mu}(r) dr \nonumber \\
	& = \int_0^{\frac{d_S^{\alpha_{SL\mu}}}{Q_T\cdot K_{SL\mu} P_S}} e^{- \Lambda'_{SL\mu}(0,Q_Tr)}e^{-\Lambda'_{ML\mu}(0,r)}\lambda'_{ML\mu}(r) dr+ \nonumber\\ & \int_{\frac{d_S^{\alpha_{SL\mu}}}{Q_T\cdot K_{ML\mu} P_S}}^{\frac{d_M^{\alpha_{ML\mu}}}{K_{SL\mu}P_M}} e^{-\Lambda'_{SL\mu}\left(0,{\frac{d_S^{\alpha_{SL\mu}}}{K_{SL\mu}P_S}}\right)} e^{-\Lambda'_{ML\mu}(0,r)}\lambda'_{ML\mu}(r) dr  \nonumber  \\
	& = \frac{1}{1 + K_1} (1 - e^{-(K_1 + 1)t_1}) + e^{-\Lambda'_{SL\mu}\left(0,{\frac{d_S^{\alpha_{SL\mu}}}{K_{SL\mu}P_S}}\right)}\left[ \exp\left(-\Lambda'_{ML\mu}\left(0,\frac{d_S^{\alpha_{SL\mu}}}{Q_TK_{SL\mu}P_S}\right)\right) - \right. \nonumber \\ & \hspace*{6cm}\left.\exp\left(-\Lambda'_{ML\mu}\left(0,\frac{d_M^{\alpha_{ML\mu}}}{K_{ML\mu}P_M}\right)\right) \right] \nonumber \label{eq: rmlgrsl},
	\end{align} where, $K_1 = \pi \lambda_S (\frac{K_{SL\mu}P_SQ_T}{P_M})^{\frac{2}{\alpha_{SL\mu}}}(\pi\lambda_M)^{-\frac{\alpha_{ML\mu}}{\alpha_{SL\mu}}}$ and $t_1 = \pi \lambda_M (K_{ML\mu}P_M)^{\frac{2}{\alpha_{ML\mu}}}\left(\frac{d_S^{\alpha_{SL\mu}}}{Q_TK_{SL\mu}P_S}\right)^{\frac{2}{\alpha_{ML\mu}}}$.
\begin{align}
\mbox{Similarly, \quad}\mathbb{P}(P_{MN\mu1}>Q_T \cdot P_{SN\mu1})  
= \exp\left(-\Lambda'_{SN\mu}\left(0,\frac{d_S^{\alpha_{SN\mu}}}{K_{SN\mu}P_S}\right)\right)  \frac{e^{-(K_2 + 1)t_2}}{1 + K_2},\nonumber
\end{align}
where $K_2 = \pi \lambda_S (\frac{K_{SN\mu}P_SQ_T}{K_{MN\mu}P_M})^{\frac{2}{\alpha_{SN\mu}}}(\pi\lambda_M)^{-\frac{\alpha_{MN\mu}}{\alpha_{SN\mu}}}$ and $t_2 = \pi\lambda_Md_M^2(K_{MN\mu}P_M)^{\frac{2}{\alpha_{MN\mu}}-1}$. 
\begin{align}
\mbox{Finally, \quad }\mathbb{P}( Q_T\cdot  P_{SL\mu1} > P_{ML\mu1}) &= 1 - \mathbb{P}(P_{ML\mu1} > Q_T\cdot  P_{SL\mu1});\nonumber\\
 \mathbb{P}(Q_T \cdot P_{SN\mu1} > P_{MN\mu1}) &= 1 -\mathbb{P}(P_{MN\mu1}>Q_T \cdot P_{SN\mu1}). \nonumber
\end{align}
Using these expressions in Eq. \eqref{eq:LOSASSO} and Eq. \eqref{eq:NLOSASSO} completes the proof.

	\section{RAT Selection Probability}
	\label{App:RATSel}
	The power received from strongest SBS of state $v$ is $P_{Sv\mu1} = (\xi_{Sv\mu1})^{-1} = K_{Sv\mu}P_{S}||x_{Sv1}||^{-\alpha_{Sv\mu}}.$
    
	So, the estimate of the mm-wave power is:
   $ P_{Svm1}= G_0 K_{SLm} P_S ||x_{Sv1}||^{-\alpha_{Svm}}$.
    Therefore the probability of sub-6GHz service, given that the user is associated with strongest SBS of visibility state $v$, is calculated as:
	\begin{align}
	\mathbb{P}_{v\mu}  &= \mathbb{P}(P_{Sv\mu1} > Q_R\times P_{Svm1}) = \mathbb{P}\left(||x_{Sv1}|| \geq \left(\frac{K_{Svm}G_0Q_R}{K_{Sv\mu}}\right)^{\frac{1}{\alpha_{Svm} - \alpha_{Sv\mu}}}\right) \nonumber \\
    & = \exp\left(-\pi\lambda_S\left(\frac{K_{Svm}G_0Q_R}{K_{Sv\mu}}\right)^{\frac{2}{\alpha_{vm} - \alpha_{v\mu}}}\right)
	\end{align}
    The probability of mm-wave service is given by $\mathbb{P}_{Svm} = 1 - \mathbb{P}_{Sv\mu}.$	This completes the proof.
	\section{Proof of Eq. \eqref{MLCovP}}
	\label{App:ProofCovP}
    We provide the derivation only for the LOS MBS association case. The other cases follow similarly. When the user is associated with the strongest LOS MBS, it experiences interference from the other LOS MBSs, the NLOS MBSs, and the SBSs. Thus, the instantaneous SINR is:  
	\begin{equation}
	SINR_{ML\mu} = \frac{h_{\xi_{ML\mu1}}(\xi_{ML\mu1})^{-1}}{I_{ML\mu} + I_{MN\mu} + I_{SL\mu} + I_{SN\mu} + \sigma_N^2}, \nonumber
	\end{equation}
	where $I_{\{.\}}$ denote the interference terms given as   
 \begin{align}I_{ML\mu} &= \sum\limits_{\xi_{ML\mu i} \in \phi'_{ML\mu}\backslash\{\xi_{ML\mu1}\}}h_{\xi_{ML\mu i}}(\xi_{ML\mu i})^{-1}; 
	\quad\quad &I_{MN\mu} &= \sum\limits_{\xi_{MN\mu i} \in \phi'_{MN}}h_{\xi_{MN\mu i}}(\xi_{MN\mu i})^{-1};\nonumber \\
	I_{SL\mu} &= \sum\limits_{\xi_{SL\mu i} \in \phi'_{SL\mu}}h_{\xi_{SL\mu i}}(\xi_{SL\mu i})^{-1};
	\quad\quad &I_{SN\mu} &= \sum\limits_{\xi_{SN\mu i} \in \phi'_{SN\mu}}h_{\xi_{SN\mu i}}(\xi_{SN\mu i})^{-1}\nonumber.
	\end{align}
\begin{flalign}
&\mbox{Now, \quad }\mathbb{P}_{CML\mu} = \mathbb{P}(SINR_{ML\mu}>\gamma)  = \mathbb{P}\left(\frac{h_{\xi_{ML\mu1}}(\xi_{ML\mu1})^{-1}}{I_{ML\mu} + I_{MN\mu} + I_{SL\mu} + I_{SN\mu} + \sigma_N^2} > \gamma \right) \nonumber \\
	& = \mathbb{P}\left(h_{\xi_{ML\mu1}} > \frac{\gamma(I_{ML\mu} + I_{MN\mu} + I_{SL\mu} + I_{SN\mu} + \sigma_N^2)}{(\xi_{ML\mu1})^{-1}} \right) \nonumber \\
	& \stackrel{(a)}{=} \mathbb{E}_{\xi_{ML\mu1}}\left\lbrace\mathbb{E}_{\phi'_{ML\mu}}\left[\exp\left(-\frac{\gamma\cdot I_{ML\mu}}{(\xi_{ML\mu1})^{-1}}\right)\right]\mathbb{E}_{\phi'_{MN\mu}}\left[\exp\left(-\frac{\gamma\cdot I_{MN\mu}}{(\xi_{ML\mu1})^{-1}}\right)\right] \right.\nonumber \\ 
	& \left.  \mathbb{E}_{\phi'_{SL\mu}}\left[\exp\left(-\frac{\gamma\cdot I_{SL\mu}}{(\xi_{ML\mu1})^{-1}}\right)\right] \mathbb{E}_{\phi'_{SN\mu}}\left[\exp\left(-\frac{\gamma\cdot I_{SN\mu}}{(\xi_{ML\mu1})^{-1}}\right)\right]\left(\exp\left(-\frac{\gamma\cdot \sigma_N^2}{(\xi_{ML\mu1})^{-1}}\right)\right)\right\rbrace \label{eq:PCMLmutemp},
	\end{flalign}
	where (a) comes from the pdf of $h_{\xi_{ML\mu1}}$. Now, 
	\begin{align}
	&\mathbb{E}_{\phi'_{ML\mu}}\left[\exp\left(-\frac{\gamma\cdot I_{ML\mu}}{(\xi_{ML\mu1})^{-1}}\right)\right]  = \mathbb{E}\left[\exp\left(-\frac{\gamma\cdot \sum\limits_{\phi'_{ML\mu}\backslash\{\xi_{ML\mu 1}\}}h_{y}{y^{-1}}}{(\xi_{ML\mu i})^{-1}}\right)\right] \nonumber \\
	& = \mathbb{E}\left[\prod\limits_{\phi'_{ML\mu}\backslash\{\xi_{ML\mu1}\}} \mathbb{E}_{h_{y}}\left[\exp\left(-\frac{\gamma\cdot h_{y}{(y)^{-1}}}{{(\xi_{ML\mu1})^{-1}}}\right)\right]\right] \nonumber\\
	& = \exp \left(- \int\limits_{\xi_{ML\mu1}}^\infty\left(1 - \mathbb{E}_{h_y} \left[ \exp\left(-\frac{\gamma\cdot h_{y}y^{-1}}{{(\xi_{ML\mu1})^{-1}}}\right)\right] \right)\Lambda'_{ML\mu}(dy)\right) \nonumber \\
	&= \exp \left(- \int\limits_{\xi_{ML\mu1}}^\infty\left(\frac{\gamma\xi_{ML\mu1}}{y + \gamma\xi_{ML\mu1}} \Lambda'_{ML\mu}(dy) \right)\right).\nonumber \\
	\mbox{Similarly, \quad }&\mathbb{E}_{\phi'_{tv\mu}}\left[\exp\left(-\frac{\gamma\cdot I_{tv\mu}}{(\xi_{ML\mu1})^{-1}}\right)\right] = \exp \left(- \int\limits_{l_{tv}}^\infty\left(1 - \frac{y}{y + \gamma\xi_{ML\mu1}} \Lambda'_{tv\mu}(dy) \right)\right), \nonumber
\nonumber
	\end{align}
    for $tv = MN,SL$ and $SN$, respectively, where the lower indexes are:  $l_{SL} = l_{SN} = Q_T\cdot \xi_{ML\mu 1}$ and $l_{MN} = \xi_{ML\mu 1}$.
	Substituting the above results in Eq. \eqref{eq:PCMLmutemp}, and taking the expectation with respect to $\xi_{ML\mu1}$, completes the proof.

    \section{Proof of Proposition 1}
    \label{App:Same}
{Consider two LOS SBS $S_1$ and $S_2$\footnote{The analysis where there are NLOS SBS can be performed with similar reasoning.}. Let the power received by the typical user from the SBS $S_1$ in mm-wave and sub-6Ghz band be $P_{S1m}$ and $P_{S1\mu}$, respectively. Let the corresponding values for $S_2$ be $P_{S2m}$ and $P_{S2\mu}$, respectively. 
Now
\begin{eqnarray}
P_{S1\mu} \geq P_{S2\mu} & \iff & K_{\mu}P_Sd_1^{\alpha_{Sv\mu}} \geq K_{\mu}P_Sd_2^{\alpha_{Sv\mu}} \nonumber \\
& \iff & K_{m}P_Sd_1^{\alpha_{Svm}} \geq K_{m}P_Sd_2^{\alpha_{Svm}} \nonumber \\
& \iff & P_{S1m} \geq P_{S2m} \nonumber \\
& \iff & Q_R P_{S1m} \geq Q_RP_{S2m} 
\label{eq:ordering}
\end{eqnarray}
}    \section{Probability of Sub-Optimal Association}
    \label{App:Sub_OPT}
  {  Recall that $E_1$ and $E_2$ denote the events the biased received power from the strongest SBS (denoted $S_1$) in sub-6GHz band is less than that received from the strongest MBS (denoted by $M_1$) and the biased received power from $S_1$ in mm-wave is higher than the received power from $M_1$, respectively. }   
   {We have:
\begin{align}
\mathbb{P}\left[E_2\mathrel{\stretchto{\mid}{3ex}} E_1\right] &= \frac{\mathbb{P}\left[E_2 \cap E_1\right]}{\mathbb{P}\left[E_1\right]}  = \frac{\mathbb{P}\left[K_m P_S Q_R Q_T G_0 d_{S1}^{-\alpha_{Svm}} \geq K_\mu P_M d_{M1}^{-\alpha_{Mv'\mu}} \cap  P_Md_{M1}^{-\alpha_{Mv'\mu}} \geq Q_T P_S d_{S1}^{-\alpha_{Sv\mu}} \right]}{\mathbb{P}\left[P_Md_{M1}^{-\alpha_{Mv'\mu}} \geq Q_T P_S d_{S1}^{-\alpha_{Sv\mu}}\right]} \nonumber \\
& = \frac{\mathbb{P}\left[d_{S1} < \left(\frac{K_m P_S Q_R Q_T G_0}{K_\mu P_M}d_{M1}^{\alpha_{Mv'\mu}}\right)^{\frac{1}{\alpha_{Svm}}} \cap d_{S1} \geq \left(\frac{P_SQ_T}{P_M} d_{M1}^{\alpha_{Mv'\mu}}\right)^{\frac{1}{\alpha_{Sv\mu}}} \right]}{\mathbb{P}\left[d_{S1} \geq \left(\frac{P_SQ_T}{P_M} d_{M1}^{\alpha_{Mv'\mu}}\right)^{\frac{1}{\alpha_{Sv\mu}}}\right]} \nonumber \\
& = \mathbb{E}_{d_{M1}}\left[\frac{\exp\left(-\pi \lambda_S \left(\left(\zeta_2x^{\alpha_{Mv'\mu}}\right)^{\frac{2}{\alpha_Svm}} - \left(\zeta_1x^{\alpha_{Mv'\mu}}\right)^{\frac{2}{\alpha_Sv\mu}}\right)\right)}{  \exp\left(-\pi \lambda_S \left(\zeta_1x^{\alpha_{Mv'\mu}}\right)^{\frac{2}{\alpha_Sv\mu}}\right)}\right] \nonumber \\
& = 2\pi\lambda_M\int_0^{d_M}\frac{\exp\left(-\pi \lambda_S \left(\left(\zeta_2x^{\alpha_{Mv'\mu}}\right)^{\frac{2}{\alpha_Svm}} - \left(\zeta_1x^{\alpha_{Mv'\mu}}\right)^{\frac{2}{\alpha_Sv\mu}}\right)\right)}{  \exp\left(-\pi \lambda_S \left(\zeta_1x^{\alpha_{Mv'\mu}}\right)^{\frac{2}{\alpha_Sv\mu}}\right)} x \exp(-\pi\lambda_M x^2)dx  \nonumber 
\end{align}
Solving this integral with the approximated values of the path-loss exponents completes the proof.}
	\bibliography{refer.bib}

\begin{thebibliography}{10}
\providecommand{\url}[1]{#1}
\def\UrlFont{\rmfamily}
\providecommand{\newblock}{\relax}
\providecommand{\bibinfo}[2]{#2}
\providecommand\BIBentrySTDinterwordspacing{\spaceskip=0pt\relax}
\providecommand\BIBentryALTinterwordstretchfactor{4}
\providecommand\BIBentryALTinterwordspacing{\spaceskip=\fontdimen2\font plus
\BIBentryALTinterwordstretchfactor\fontdimen3\font minus
  \fontdimen4\font\relax}
\providecommand\BIBforeignlanguage[2]{{%
\expandafter\ifx\csname l@#1\endcsname\relax
\typeout{** WARNING: IEEEtran.bst: No hyphenation pattern has been}%
\typeout{** loaded for the language `#1'. Using the pattern for}%
\typeout{** the default language instead.}%
\else
\language=\csname l@#1\endcsname
\fi
#2}}

\bibitem{andrews2014will}
J.~G. Andrews, \emph{et~al.}, ``{What Will 5G Be?}'' \emph{IEEE J. Sel. Areas
  Commun.}, vol.~32, no.~6, pp. 1065--1082, 2014.

\bibitem{lopez2015towards}
D.~L{\'o}pez-P{\'e}rez, \emph{et~al.}, ``{Towards 1 Gbps/UE in Cellular
  Systems: Understanding Ultra-Dense Small Cell Deployments},'' \emph{IEEE
  Commun. Surveys Tuts.}, vol.~17, no.~4, pp. 2078--2101, 2015.

\bibitem{rappaport2013millimeter}
T.~S. Rappaport, \emph{et~al.}, ``{Millimeter Wave Mobile Communications for 5G
  Cellular: It Will Work!}'' \emph{IEEE Access}, vol.~1, pp. 335--349, 2013.

\bibitem{okino2011pico}
K.~Okino, \emph{et~al.}, ``{Pico Cell Range Expansion with Interference
  Mitigation toward LTE-Advanced Heterogeneous Networks},'' in \emph{IEEE ICC
  Workshops}, 2011, pp. 1--5.

\bibitem{eguizabal2013interference}
M.~Eguizabal and A.~Hernandez, ``Interference management and cell range
  expansion analysis for {LTE} picocell deployments,'' in \emph{IEEE PIMRC},
  2013, pp. 1592--1597.

\bibitem{ali2015load}
M.~S. Ali, P.~Coucheney, and M.~Coupechoux, ``Load balancing in heterogeneous
  networks based on distributed learning in potential games,'' in \emph{IEEE
  WiOpt}, 2015, pp. 371--378.

\bibitem{ghosh2014millimeter}
A.~Ghosh, \emph{et~al.}, ``{Millimeter-Wave Enhanced Local Area Systems: A
  High-Data-Rate Approach for Future Wireless Networks},'' \emph{IEEE J. Sel.
  Areas Commun.}, vol.~32, no.~6, pp. 1152--1163, 2014.

\bibitem{li2016initial}
Y.~Li, \emph{et~al.}, ``On the initial access design in millimeter wave
  cellular networks,'' in \emph{IEEE GLOBECOM Wkshps.}, 2016, pp. 1--6.

\bibitem{mmMagic}
{H2020-ICT-671650 mmMAGIC}, ``{D3.1: Initial concepts on 5G architecture and
  integration},'' \emph{{Available Online at https://5g-mmmagic.eu/}}, Mar.
  2016.

\bibitem{kangas2013angle}
A.~Kangas and T.~Wigren, ``Angle of arrival localization in {LTE} using {MIMO}
  pre-coder index feedback,'' \emph{IEEE Commun. Lett.}, vol.~17, no.~8, pp.
  1584--1587, 2013.

\bibitem{onoe20161}
S.~Onoe, ``{Evolution of 5G mobile technology toward 2020 and beyond},'' in
  \emph{IEEE ISSCC}, 2016, pp. 23--28.

\bibitem{5783993}
Z.~Pi and F.~Khan, ``An introduction to millimeter-wave mobile broadband
  systems,'' \emph{IEEE Commun. Mag.}, vol.~49, no.~6, pp. 101--107, June 2011.

\bibitem{elsawy2013stochastic}
H.~ElSawy, E.~Hossain, and M.~Haenggi, ``{Stochastic Geometry for Modeling,
  Analysis, and Design of Multi-Tier and Cognitive Cellular Wireless Networks:
  A Survey},'' \emph{IEEE Commun. Surveys Tuts.}, vol.~15, no.~3, pp.
  996--1019, 2013.

\bibitem{bai2015coverage}
T.~Bai and R.~W. Heath, ``{Coverage and Rate Analysis for Millimeter-Wave
  Cellular Networks},'' \emph{IEEE Trans. Wireless Commun.}, vol.~14, no.~2,
  pp. 1100--1114, 2015.

\bibitem{singh2014joint}
S.~Singh and J.~G. Andrews, ``{Joint Resource Partitioning and Offloading in
  Heterogeneous Cellular Networks},'' \emph{IEEE Trans. Wireless Commun.},
  vol.~13, no.~2, pp. 888--901, 2014.

\bibitem{di2015stochastic}
M.~Di~Renzo, ``{Stochastic Geometry Modeling and Analysis of Multi-Tier
  Millimeter Wave Cellular Networks},'' \emph{IEEE Trans. Wireless Commun.},
  vol.~14, no.~9, pp. 5038--5057, 2015.

\bibitem{omar2016performance}
M.~S. Omar, \emph{et~al.}, ``Performance analysis of hybrid 5g cellular
  networks exploiting mmwave capabilities in suburban areas,'' in \emph{IEEE
  ICC}, 2016, pp. 1--6.

\bibitem{yao2016coverage}
G.~Yao, \emph{et~al.}, ``Coverage and rate analysis for non-uniform
  millimeter-wave heterogeneous cellular network,'' in \emph{8th IEEE WCSP,
  2016}, 2016, pp. 1--6.

\bibitem{elshaer2016downlink}
H.~Elshaer, \emph{et~al.}, ``{Downlink and Uplink Cell Association With
  Traditional Macrocells and Millimeter Wave Small Cells},'' \emph{IEEE Trans.
  Wireless Commun.}, vol.~15, no.~9, pp. 6244--6258, Sept. 2016.

\bibitem{bonald2003wireless}
T.~Bonald and A.~Prouti{\`e}re, ``{Wireless Downlink Data Channels: User
  Performance and Cell Dimensioning},'' in \emph{ACM MobiCom}, 2003, pp.
  339--352.

\bibitem{blaszczyszyn2014user}
B.~B{\l}aszczyszyn, M.~Jovanovicy, and M.~K. Karray, ``{How user throughput
  depends on the traffic demand in large cellular networks},'' in \emph{IEEE
  WiOpt}, 2014, pp. 611--619.

\bibitem{singh2013offloading}
S.~Singh, H.~S. Dhillon, and J.~G. Andrews, ``Offloading in heterogeneous
  networks: Modeling, analysis, and design insights,'' \emph{IEEE Trans.
  Wireless Commun.}, vol.~12, no.~5, pp. 2484--2497, 2013.

\bibitem{bai2014analysis}
T.~Bai, R.~Vaze, and R.~W. Heath, ``{Analysis of Blockage Effects on Urban
  Cellular Networks},'' \emph{IEEE Trans. Wireless Commun.}, vol.~13, no.~9,
  pp. 5070--5083, 2014.

\bibitem{36.814}
{3GPP TSG RAN}, ``{TR 36.814, E-UTRA; Further advancements for E-UTRA physical
  layer aspects},'' \emph{v9.0.0}, March 2010.

\bibitem{38.900}
------, ``{TR 38.900, Study on channel model for frequency spectrum above 6
  GHz},'' \emph{v14.1.0}, September 2016.

\bibitem{zhang2014stochastic}
X.~Zhang and M.~Haenggi, ``{A Stochastic Geometry Analysis of Inter-Cell
  Interference Coordination and Intra-Cell Diversity},'' \emph{IEEE Trans.
  Wireless Commun.}, vol.~13, no.~12, pp. 6655--6669, 2014.

\bibitem{shokri2015millimeter}
H.~Shokri-Ghadikolaei, \emph{et~al.}, ``{Millimeter Wave Cellular Networks: A
  MAC Layer Perspective},'' \emph{{IEEE Trans. Commun.}}, vol.~63, no.~10, pp.
  3437--3458, 2015.

\bibitem{ghatak2016performance}
G.~Ghatak, A.~De~Domenico, and M.~Coupechoux, ``{Performance Analysis of
  Two-tier Networks with Closed Access Small-cells},'' in \emph{{IEEE WiOpt}},
  2016, pp. 1--8.

\bibitem{leon2008probability}
A.~Leon-Garcia, \emph{Probability, statistics, and random processes for
  electrical engineering}.\hskip 1em plus 0.5em minus 0.4em\relax
  Pearson/Prentice Hall, 3rd ed., 2008.

\bibitem{7488044}
S.~M. Hasan, M.~A. Hayat, and M.~F. Hossain, ``{On the downlink SINR and outage
  probability of stochastic geometry based LTE cellular networks with
  multi-class services},'' in \emph{18th ICCIT}, Dec 2015, pp. 65--69.

\end{thebibliography}
	\bibliographystyle{IEEEtran}	

	\begin{acronym}
		\acro {BS} {base station}
		\acro {ICIC} {inter-cell interference co-ordination}
		\acro {LOS} {line of sight}
		\acro {NLOS} {non-line of sight}
		\acro {MBS} {macro eNode-B}
		\acro {SBS} {small cell eNode-B}
		\acro {RAT} {radio access technique}
	\end{acronym}
\end{document}